\documentclass[journal]{IEEEtran}
\usepackage{amsfonts}
\ifCLASSINFOpdf \else \fi

\usepackage[dvips]{graphicx}
\usepackage{subfigure}
\usepackage{enumerate}
\usepackage{algorithm}
\usepackage{algorithmic}
\usepackage{amsmath}
\usepackage{booktabs}  
\usepackage{bbding}
\usepackage{color}
\usepackage[numbers,sort&compress]{natbib}

\usepackage{amsmath}
\newtheorem{lemma}{\bf Lemma}
\newtheorem{theorem}{\bf Theorem}
\newtheorem{proof}{Proof}[section]

\begin{document}

\title{Fast Neighbor Discovery for Wireless Ad Hoc Network with Successive Interference Cancellation}

\author{Zhiqing~Wei,~\IEEEmembership{Member,~IEEE,}
        Yueyue~Liang,
        Zeyang~Meng,~\IEEEmembership{Graduate~Student~Member,~IEEE,}\\
        Zhiyong~Feng,~\IEEEmembership{Senior Member,~IEEE,}
        Kaifeng~Han,~\IEEEmembership{Member,~IEEE,}
        Huici~Wu,~\IEEEmembership{Member,~IEEE}

\thanks{		
	Zhiqing Wei, Yueyue Liang, Zeyang Meng, and Zhiyong Feng are
with the Key Laboratory of Universal Wireless Communications, Ministry of Education,
School of Information and Communication Engineering,
Beijing University of Posts and Telecommunications, Beijing 100876, China
(e-mail: \{weizhiqing, liangyue, mengzeyang, fengzy\}@bupt.edu.cn).

Huici Wu is with the National Engineering Lab for Mobile Network Technologies, Beijing University of Posts and Telecommunications, Beijing 100876, China (email: dailywu@bupt.edu.cn).
	
Kaifeng Han is with the China Academy of Information
and Communications Technology, Beijing 100191, China (e-mail: hankaifeng@caict.ac.cn).
		
Correspondence authors: Kaifeng Han, Zhiyong Feng, Huici Wu.}}

\maketitle

\begin{abstract}
Neighbor discovery (ND) is a key step in wireless ad hoc network, which directly affects the efficiency of wireless networking.
Improving the speed of ND has always been the goal of ND algorithms.
The classical ND algorithms lose packets due to the collision of multiple packets,
which greatly affects the speed of the ND algorithms.
Traditional methods detect packet collision and implement retransmission when encountering packet loss.
However, they does not solve the packet collision problem and the performance improvement of ND algorithms is limited.
In this paper, the successive interference cancellation (SIC) technology is introduced into the ND algorithms to unpack multiple collision packets by distinguishing multiple packets in the power domain.
Besides, the multi-packet reception (MPR) is further applied to reduce the probability of packet collision by distinguishing multiple received packets, thus further improving the speed of ND algorithms.
Six ND algorithms, namely completely random algorithm (CRA), CRA based on SIC (CRA-SIC), CRA based on SIC and MPR (CRA-SIC-MPR), scan-based algorithm (SBA), SBA based on SIC (SBA-SIC), and SBA based on SIC and MPR (SBA-SIC-MPR),
are theoretically analyzed and verified by simulation.
The simulation results show that SIC and MPR reduce the ND time of SBA by 69.02\% and CRA by 66.03\% averagely.
\end{abstract}
\IEEEpeerreviewmaketitle

\IEEEpeerreviewmaketitle

\section{Introduction}
Recently, Internet of Things (IoT) has been widely applied in various fields, such as military \cite{5}, agriculture \cite{6}, medical treatment \cite{7}, industry \cite{10}, smart home \cite{8} and smart city \cite{9}.
As a distributed network, wireless ad hoc network is widely applied to support IoT with convenient and flexible networking ability \cite{1, 2, 3, 4}.
However, due to the large number of nodes, fast networking of wireless ad hoc network is challenging.

As the first step of networking, neighbor Discovery (ND) directly affects the efficiency of networking \cite{11}.
Recently, the research on ND algorithms is mainly focused on reducing energy consumption and improving ND speed \cite{c1}.
This paper focuses on the improvement of ND speed.
As a key performance indicator of the speed of ND algorithms,
the ND time is defined as the number of time slots required to discover all neighbors.
To reduce ND time, the existing methods are roughly divided into the following four categories.

    \subsubsection{Exploitation of prior information}
    Without prior information, the nodes explore neighbors blindly in all directions.
    To avoid repeating ND attempts in the directions without potential neighbors,
    various sensing methods are applied to obtain the prior information of node distribution,
    which avoid invalid ND in advance \cite{12, 13, 14, a1}.
    In \cite{12} and \cite{13}, radar is applied to provide the location information of neighbors as prior information.
    The ND algorithm proposed in \cite{14} is implemented in a dual-band system,
    where the prior neighbor information obtained in one frequency band assists ND in the other frequency band.
    In \cite{a1}, roadside units detect the position of vehicles and assist vehicles to discovery their neighbors.
    Besides, positioning technique is applied to obtain the distribution of neighbors,
    which improves the efficiency of ND \cite{15}.

    \subsubsection{Optimization of parameters}
    In \cite{b1} and \cite{b2},
    reinforcement learning is introduced to find the optimal ND strategy by interacting
    with the environment.
    In \cite{b3} and \cite{b4}, using a tool of machine learning, ND is formulated as a multi-armed bandit problem.
    Nodes utilize the result of past discovery attempts for learning.
    The methods in \cite{b5, b6, b7, b8} make ND adaptively based on number of neighbors \cite{b5},
    position accuracy \cite{b6}, channel randomness \cite{b7}, or collision probability \cite{b8}.

    \subsubsection{Reduction of packet collision probability}
    Packets that collide are dropped, resulting in longer ND time.
    To speed up ND, various schemes to reduce the packet collision rate are proposed.
    In \cite{16}, Liu \textit{et al.} introduce a third state, namely, the idle state.
    The nodes in idle state do not receive or transmit signals,
    which will reduce the probability of packet collisions and energy consumption.
    In \cite{a2,a3}, ND with multi-channel capability effectively reduces packet collisions in the network
    especially for dense networks. In \cite{17}, Zhao \textit{et al.} propose a 3-way multi-carrier asynchronous ND algorithm,
    which verifies that the collision probability of a multi-carrier system is lower than that of a single carrier system.
    Besides, the stop mechanism is applied as soon as the handshake is accomplished to reduce the collision probability.
    In \cite{a4}, Liu \textit{et al.} present a new anti-collision strategy for ND named Dual Channel Competition (DCC),
    in which two time slots are used as a time frame.
    The first time slot is called competition slot (CS) and the second one is called message slot (MS).
    In CS, nodes compete for the authority of transmission in MS, reducing packet collisions in MS.

    \subsubsection{Conflict detection and retransmission}
    In addition to the above-mentioned methods for avoiding packet collisions,
    a collision resolution scheme is proposed in \cite{18, a5}, which retransmits immediately when a collision occurs.
    If the energy detector detects that a collision occurs,
    the receiving node will send a collision acknowledgment to transmitting nodes.
    Then the receiving node will switch to the ``conflict resolution listening mode'',
    and transmitting nodes will switch to the ``conflict resolution retransmission mode'' until the data packets
    from at least two transmitting nodes are successfully received.

    The above methods have paid much attention to collision, which is a key factor that leads to the prolonged ND time.
    However, the above methods did not study the recovery of collided packets at the receiving node.
    To address the collision problem during ND,
    we apply the successive interference cancellation (SIC) technology to recover packets when collision occurs to reduce ND time.
    This scheme can be applied based on the proposed neighbor discovery algorithms.
    Although the introduction of SIC reduces the ND time, additional signal processing is required, which improves the complexity.
    Besides, the multi-packet reception (MPR) is applied to avoid packet collision,
    which further improves the performance of ND algorithms.
    The main contributions of this paper are as follows.

\begin{enumerate}
\item The ND algorithms with SIC are proposed to avoid packet collisions,
thereby reducing ND time and improving ND performance.
Simulation results show that compared with traditional scan-based neighbor discovery algorithm (SBA)
and completely random neighbor discovery algorithm (CRA),
SIC enabled SBA and CRA can be improved by an average of 34.03\% and 22.38\%.
\item The ND time of the proposed algorithm is theoretically analyzed.
We discover that there exists an upper bound of the number of packet collisions that SIC can handle.
Besides, the ability of MPR to avoid data packet collision mainly depends on the number of modulation methods.
Too few modulation methods still cannot avoid data packet collisions well,
and too many modulation methods result in wasted resources.
Therefore, the number of modulation methods needs to be selected.
\end{enumerate}

It is noted that part of this paper was our previous work as a conference paper \cite{wcsp}.
Compared with the conference paper,
this paper has the following improvements.

\begin{enumerate}
\item This paper not only applies SIC to achieve collided packets recovery,
but also applies MPR for collision avoidance, which comprehensively solves the problem of data packet collision.
This ability of SIC and MPR is further applied to the ND algorithms to obtain more optimal ND schemes.
\item This paper provides detailed design and performance analysis of the SIC enabled ND algorithms.
A description of the ND process and the derivation of the expectations of ND time are provided.
In addition, the practical issues in SIC, such as distance constraints of SIC and imperfect SIC,
are considered and analyzed in this paper.
\item The optimal parameters minimizing ND time,
the impact of imperfect SIC on ND time,
and the performance improvement of ND algorithms are obtained through simulation.
\end{enumerate}

Other sections are organized as follows.
Some related works on ND algorithms and SIC technology are introduced in Section II.
Section III and Section IV provide a detailed description of the network model and the process of ND, respectively.
Section V introduces SIC and MPR and provides a theoretical analysis of the ND algorithms, namely, CRA, SBA, CRA-SIC, SBA-SIC, CRA-SIC-MPR, and SBA-SIC-MPR.
The probability that a node successfully discovers a neighbor and the time expectation of completing the ND process in the ND algorithms are obtained.
In Section VI, simulation results and analysis for the ND algorithms are provided.
Finally, Section VII summarizes this paper.
Table I lists the main notations and their descriptions in this paper.

\begin{table}[htbp]
    \newcommand{\tabincell}[2]{\begin{tabular}{@{}#1@{}}#2\end{tabular}}
	\centering
	\caption{Main notations}
	\begin{tabular}{cl}
		\toprule  
		Notation&Description \\
		\midrule  
        $\rm{CRA}$&Completely random algorithm\\[1.5ex]
        $\rm{SBA}$&Scan-based algorithm\\[1.5ex]
        $\rm{SIC}$&Successive interference cancellation\\[1.5ex]
        $\rm{MPR}$&Multi-packet reception\\[1.5ex]
		$\theta $&Beam width of the scanning beam \\[1.5ex]
        $a$&Length of nodes distribution \\[1.5ex]
        $b$&Width of node distribution \\[1.5ex]
        $\lambda $&Distribution density of nodes \\[1.5ex]
        $r$&Communication radius of nodes \\[1.5ex]
        $\bar N$&The average number of neighbors of a node \\[1.5ex]
        $K$&The average number of neighbors of a node in a beam \\[1.5ex]
        ${P_t}$&Transmit probability \\[1.5ex]
        $D(t)$&\tabincell{l}{The number of neighbors that discovered the node in the\\past $t$ time slots in a beam}\\[2.5ex]
        $P_{{\rm{A}} \to {\rm{B}}}^{x}\left( t \right)$&\tabincell{l}{The probability that node A finds its unknown neighbor\\B in the $t$-th time slot based on the $x$ algorithm, where \\
        $\begin{array}{l}
        x \in \{ \rm{CRA,}\;\rm{SBA,}\;\rm{CRA\_SIC,}\;\rm{SBA\_SIC,}\;\\
        \;\;\;\;\;\;\;\;\rm{CRA\_SIC\_MPR,}\;\rm{SBA\_SIC\_MPR}\}
        \end{array}$}\\[4.5ex]
        $M$&The number of data packets collision \\[1.5ex]
        $\beta $&SIR or SINR threshold for successful unpacking \\[1.5ex]
        ${S_i}$&\tabincell{l}{The power of the $i$-th data collided packet}\\[1.5ex]
        ${d_i}$&\tabincell{l}{The distance between the transmitting node of the $i$-th\\collided packet and the receiving node} \\[2.5ex]
        ${\lambda _0}$&Free space wavelength \\[1.5ex]
        ${P_T}$&Transmit power of nodes \\[1.5ex]
        ${G_T}$&Transmission gain of nodes\\[1.5ex]
        ${G_R}$&Reception gain of nodes\\[1.5ex]
        ${n_0}$&\tabincell{l}{The maximum number of data packets that can be\\unpacked by perfect SIC at the same time} \\[2.5ex]
        $P(Q,M)$&\tabincell{l}{The probability that the first $Q$ of the $M$ collided data\\ packets can be successfully unpacked} \\[2.5ex]
        $\bar P(1,M)$&\tabincell{l}{The expected probability that one of $M$ collided packets\\ can be successfully unpacked} \\[2.5ex]
        ${N_0}$&\tabincell{l}{The power of additive white Gaussian noise in the\\environment}\\[2.5ex]
        $\xi $&Residual coefficient of interference cancellation \\[1.5ex]
        ${N_i}$&\tabincell{l}{The power of additive white Gaussian noise caused by\\imperfect cancellation} \\[2.5ex]
        $N$&\tabincell{l}{The sum power of noise at the receiver and noise caused\\by imperfect cancellation}\\[2.5ex]
        ${C_i}$&The power of residual of interference cancellation \\[1.5ex]
        $E\left( {{T_{all}}} \right)$&\tabincell{l}{The number of time slots expected by node A to find all\\neighbors in all beams} \\[2.5ex]
        $h$&The number of modulation methods \\
		\bottomrule  
	\end{tabular}
\end{table}

\section{Related Works}

    \subsection{Neighbor discovery}
    According to the handshake rules,
    ND algorithms consist of the algorithms with one-way handshake, two-way handshake, and three-way handshake
    according to different handshake rules.
    For one-way handshake, the node only needs to mark that a neighbor has been discovered when the signal is correctly received.
    One-way handshake is mostly used in the scenario with omni-directional antennas with the characteristics of
    omni-directional transmission and reception \cite{19}.
    Unlike omni-directional antennas, when nodes use directional antennas to transmit and receive signals with neighbors,
    beam alignment is required \cite{20}.
    Consequently, in the scenario with directional antennas,
    nodes adopt two-way handshake \cite{21} or three-way handshake \cite{17,22}.
    Since this paper adopts directional antennas,
    we adopt the two-way handshake mechanism in ND algorithms.

    According to the channel access mechanism, ND algorithms consist of CRA and SBA \cite{20}.
    The difference between CRA and SBA is the scan order of beams.
    For SBA, all nodes have the same beam scanning order sequence.
    Subsequently, it is necessary to select transmission or reception according to the predefined beam on each time slot.
    For CRA, the beam scanning order of nodes in each time slot is completely random.
    Nodes will randomly select one beam for transmission or reception with equal probability in a time slot.

    As the first step of networking, ND directly affects the performance of routing protocols.
    Therefore, the research of ND is also carried out in some routing protocols.
    In \cite{b9}, Oubbati \textit{et al.} comprehensively considered the balanced energy consumption,
    the link breakage prediction, and the connectivity degree to minimize the number of path failures,
    decrease the packet losses, and increase the lifetime of the network.
    In \cite{b10}, a routing protocol based on random network coding and clustering is
    designed to reduce the number of hops in routing protocol.

    Table II provides a comparison among the previously described ND algorithms with our ND algorithms.

\begin{table*}[htbp]
    \newcommand{\tabincell}[2]{\begin{tabular}{@{}#1@{}}#2\end{tabular}}
	\centering  
	\caption{Comparison of the related ND algorithms}  

	\begin{tabular}{|c|c|c|c|c|c|c|c|}
		\hline  
		& & & & & & & \\[-6pt]
		&Our algorithms &RCI-SBA \cite{12} &Rns \cite{13} &GSIM-ND \cite{a1} &MC-NDA \cite{17} &3D-ND \cite{21} &HAS-3-way \cite{22}\\  
		\hline
		& & & & & & & \\[-6pt]  
		Handshake rules &two-way &two-way &three-way &two-way &three-way &two-way &three-way \\
\hline
		& & & & & & & \\[-6pt]  
        Channel access mechanism &CRA\verb|\|SBA &SBA &CRA &CRA &CRA &SBA &CRA \\
        \hline
		& & & & & & & \\[-6pt]  
        Node size &Dense &Sparse &Sparse &Dense &Dense &Sparse &Dense \\
        \hline
        & & & & & & & \\[-6pt]  
		Time synchronization &\checkmark &\checkmark &\checkmark &\checkmark & &\checkmark &\checkmark \\
        \hline
		& & & & & & & \\[-6pt]  
        Prior information & &\checkmark &\checkmark &\checkmark & & & \\
        \hline
		& & & & & & & \\[-6pt]  
        Reduce energy consumption & &\checkmark & & & &\checkmark &\checkmark \\
        \hline
		& & & & & & & \\[-6pt]  
        Multi-packet reception &\checkmark & & &\checkmark &\checkmark & & \\
        \hline
		& & & & & & & \\[-6pt]  
        Packet collision resolution &\checkmark & & & & & & \\
		\hline
	\end{tabular}
\end{table*}

    \subsection{Successive interference cancellation}

    SIC techniques enable multi-packet reception,
    which detects each data packet through an iterative method.
    SIC was first known as an application in Code Division Multiple Access (CDMA) \cite{23},
    which is further utilized in Non Orthogonal Multiple Access (NOMA) \cite{24}.
    This paper applies SIC to enhance the efficiency of ND algorithms.

    The execution process of SIC when receiving signals consists of the following four steps.

\begin{enumerate}[Step 1:]
\item Preprocess and sort signals: When receiving multiple data packets, the receiver sends them to the matched filter for pre-processing and sorts the pre-processed signals according to the signal strength from strong to weak.
\item Unpack the first data packet: According to the sorting results, the first data packet is obtained through the procedures of detection, judgment, and reconstruction.
\item Remove the interference: The first signal that has been unpacked is subtracted from received signals to eliminate the interference of this signal on subsequent signals.
\item Repeat iteration: Repeat the above steps to recover the remaining signals.
\end{enumerate}

\section{Network Model}

    In the ND algorithms with directional antenna,
    it is assumed that the transmission power and communication range of all nodes are the same \cite{19,12}.
    The nodes in the entire network are located on the same horizontal plane and uniformly distributed with a density of $\lambda $ nodes per unit area \cite{c2}.
    The nodes use the directional antenna.
    As shown in Fig. \ref{3}, the nodes are covered by the beams of $\frac{{2\pi }}{\theta }$ sector antenna elements in all directions,
    where $\theta \;(0 < \theta  < 2\pi )$ represents the beam width of each sector antenna element.
    By selecting different sector antenna elements, the nodes can switch between sectors.
    The nodes are divided by sector antenna elements in space dimension and by time slots in time dimension.
    This paper assumes that all nodes are synchronized, and similar assumptions are made in \cite{b5,20,c3}.
    Synchronization could be achieved by means in \cite{c4,c5}.
    In addition, detailed descriptions of the ND algorithms are as follows.

\begin{enumerate}
\item The ND algorithms adopt two-way handshaking mechanism.
Thus, each time slot is divided into two equal mini-slots, as shown in Fig. \ref{4}.
\item The communication between nodes is half-duplex.
Nodes transmit with probability ${P_t}$ or receive with probability $1 - {P_t}$ at the first mini-slot.
\item To reduce the probability of data packet collisions,
all nodes follow ``stop as soon as handshake accomplished'' mechanism,
which means that nodes will no longer reply to neighbors that have been discovered.
\item In the single packet reception algorithm, if the received signals come from two or more neighbors,
the node will directly drop the data packets because the packet collision occurs in the algorithm without SIC.
However, in the algorithm with SIC, nodes apply SIC technology to unpack collided data packets.
If the difference in the power of the received signals meets the unpacking threshold condition of SIC,
the data packets can be successfully recovered at the same time.
The difference in the power of the received signals depends on the location of the nodes in the collision,
which may not meet the unpacking threshold of SIC.
As a result, unpacking fails.
\item In the ND algorithms that support multi-packet reception,
the node can choose different modulation methods when transmitting the signal.
If a node receives multiple packets with different modulation methods simultaneously,
the node can unpack these packets successfully.
Otherwise, the node still tries to unpack the packets by SIC.
\end{enumerate}

\begin{figure}[!t]
\centering
\begin{minipage}[t]{0.22\textwidth}
\centering
\includegraphics[width=4cm]{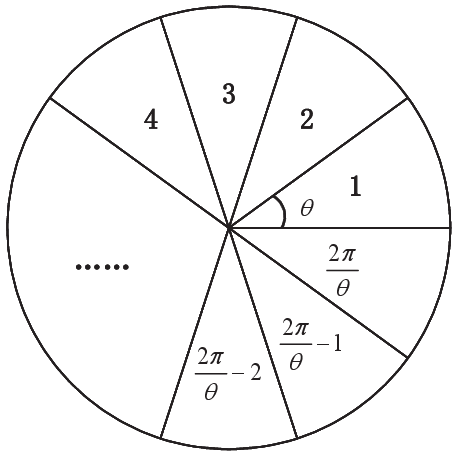}
\caption{The beams of idealized sector antennas.}
\label{3}
\end{minipage}
\hspace{4mm}
\begin{minipage}[t]{0.23\textwidth}
\centering
\includegraphics[width=4cm]{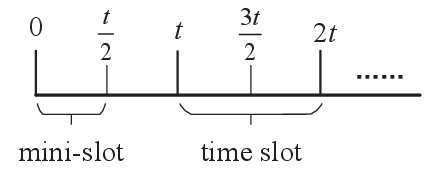}
\caption{Division of time slots.}
\label{4}
\end{minipage}
\end{figure}

\section{The Process of Neighbor Discovery}
        \subsection{CRA and SBA neighborhood discovery}
         For CRA, the beam scanning order of nodes in each time slot is completely random.
         Nodes randomly select one in $\frac{{2\pi }}{\theta }$ beams for transmission or reception with the probability of $\frac{\theta }{{2\pi }}$.
         The difference between SBA and CRA is only in scanning order of beams.
         For SBA, all nodes have the same beam scanning order sequence defined at the beginning of neighbor discovery.
         As shown in Fig. \ref{3}, nodes scan beams counterclockwise starting from beam 1.
         Subsequently, it is necessary for each node to select the beam for operation according to the predefined sequence on each time slot.
         In addition, in SBA, to ensure the transmitting and receiving directions of the nodes are opposite,
         the nodes in transmitting state transmit in the predefined beam, while the nodes in receiving state receive in the opposite direction of the predefined beam.
         As shown in Fig. \ref{5a}, CRA and SBA can successfully discover neighbors only
         when the data packets do not collide. Their specific processes are as follows.

        \begin{enumerate}[Step 1:]
        \item In the first mini-slot, node A sends a signal in the selected direction with probability ${P_t}$, and node B waits to receive a signal in the selected direction with probability $1 - {P_t}$.
        \item Node B receives the signal sent by node A and determines whether the signal is sent from a neighbor that has been discovered.
        \item If the received signal comes from a neighbor discovered by node B, node B does not reply and waits for the end of this time slot. Otherwise, Step 4 is executed.
        \item In the second mini-slot, node B switches to the transmitting state and replies with an acknowledge signal in the direction of the received signal.
        \item Node A switches to the receiving state and receives the acknowledge signal from node B, thus completing the ND process between nodes A and B.
        \end{enumerate}

        For the situation in Fig. \ref{5b}, in the first mini-slot,
        both node A and node C send signals to node B,
        and data packet collision will occur at node B.
        Since SBA and CRA cannot simultaneously unpack multiple data packets,
        Node B will discard the received data packets.
        In this situation, these three nodes cannot successfully discover neighbors.

        \subsection{SIC based neighbor discovery}
        Compared with traditional CRA and SBA, the improvement of the SIC based ND algorithms,
        i.e. CRA based on SIC (CRA-SIC) and SBA based on SIC (SBA-SIC),
        is that the multiple collided data packets can be successfully unpacked,
        such that the efficiency of ND is improved.
        As illustrated in Fig. \ref{5b}, both node A and node C are sending data packets to node B.
        The distance between node A and node B is not equal to the distance between node C and node B,
        such that the power of the data packets of node A and node C received by node B is different.
        Therefore, node B takes advantage of this power difference to unpack the data packets of node A and node C by SIC.
        The specific process is as follows.

        \begin{enumerate}[Step 1:]
        \item In the first mini-slot, node B receives two data packets from both node A and node C.
        \item Node B sorts the two received signals according to the power level.
        \item Node B regards the signal with low power (sent by node C) as interference and unpacks the signal with high power (sent by node A). It is required that the power of these two data packets received by node B has a significant power difference greater than the threshold of SIC (more details are in Section V).
        \item When node B successfully unpacks the signal from node A, it subtracts this signal from the received signals to cancel the interference to the signal transmitted by node C.
        \item Repeat Step 3 to complete the unpacking of the data packet from node C.
        \item In the second mini-slot, node B switches to the transmitting state and replies with an acknowledge signal in the direction of the received signal.
        \item Nodes A and C receive the acknowledge signal from node B. In this time slot, the ND process of both A-B and C-B is completed.
        \end{enumerate}

        \subsection{SIC and MPR based neighbor discovery}
        The improvement of SIC and MPR based ND algorithms, i.e. CRA based on SIC and MPR (CRA-SIC-MPR) and SBA based on SIC and MPR (SBA-SIC-MPR),
        is that in the case of multiple packets received simultaneously,
        the SIC and MPR based ND algorithms can successfully unpack the received packets if different modulation methods are selected for these packets.
        If the received packets have the same modulation, the receiving node can unpack the packets using the SIC method,
        and its specific process is as follows.

        \begin{enumerate}[Step 1:]
        \item In the first mini-slot, multiple nodes sending signals to node B randomly select modulation methods for the packets to be sent.
        \item When receiving signals from multiple nodes, node B determines whether the received signals adopt the same modulation method and unpacks the signals of different modulation methods.
        \item Node B determines whether there are multiple signals of the same modulation method and applies SIC to unpack them.
        \item In the second mini-slot, node B selects any modulation method to process the acknowledge signal and replies in the direction of the received signal.
        \item Multiple transmitting nodes receive the acknowledge signal from node B. In this time slot, the ND process between multiple nodes and node B is completed.
        \end{enumerate}

\begin{figure}[!t]
\centering
\subfigure[Successful communication between node A and node B]{
\includegraphics[width=3.7cm]{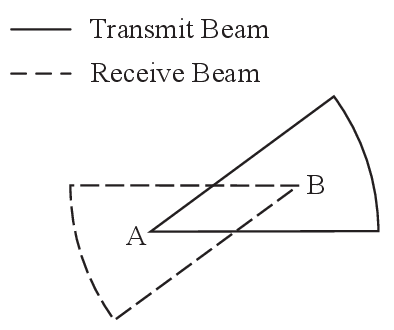}
\label{5a}
}
\hspace{4mm}
\subfigure[Collision between node A and node C]{
\includegraphics[width=3.5cm]{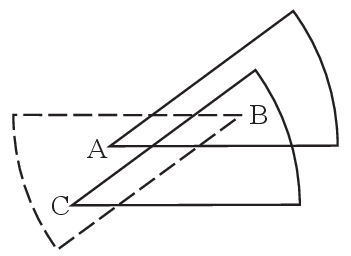}
\label{5b}
}
\caption{Successful communication and collision between neighbors.}
\end{figure}

\section{Analysis of Proposed Neighbor Discovery Algorithms}

    This section theoretically analyze the six ND algorithms:
    CRA, SBA, CRA-SIC, SBA-SIC, CRA-SIC-MPR and SBA-SIC-MPR.
    The probability of successful ND and the expectation of the time slots
    required to complete ND with the six ND algorithms are derived.

    \subsection{CRA and SBA}
    The discovery probabilities of CRA and SBA algorithms are derived in this section,
    which is different from \cite{a6} on three aspects.
    1) The average number of neighbors of the node is derived and the boundary is taken into account.
    2) When deriving the discovery probability, \cite{a6} only considers bidirectional links,
    while this paper also considers the case that only node A discovers node B, as shown in Fig. \ref{5aa}.
    3) When deriving the probability that node A receives the reply from node B without interference,
    \cite{a6} only considers the case that other nodes are in the transmitting state and
    ignores the case that other nodes are in the receiving state but fail to receive.
    In this section, all cases are considered.

    \begin{lemma}
    For the whole network, the average number of neighbors of a node is (proof can be found in Appendix A)
    \begin{equation}
    \label{eq_L1}
    \bar N = \frac{{3\pi \lambda {r^4} - 8\left( {a + b} \right)\lambda {r^3} + 6\pi \lambda ab{r^2}}}{{6ab}}
    \end{equation}
    and the average number of neighbors of a node in a beam is
    \begin{equation}
    K = \frac{\theta }{{2\pi }}\bar N,
    \end{equation}
    where $a$ and $b$ are the length and width of the area of node distribution respectively,
    $r$ is the communication radius of nodes,
    $\lambda$ is the distribution density of nodes and $\theta $ is the beam width of the scanning beam.
    \end{lemma}

    \begin{theorem}
    The probability of node A discovering the neighboring node B at $t$-th time slot is
    \begin{equation}
    \normalsize
    \begin{array}{l}
    P_{{\rm{A}} \to {\rm{B}}}^{CRA}\left( t \right) = \frac{\theta }{{2\pi }}{P_t} \cdot \frac{\theta }{{2\pi }}\left( {1 - {P_t}} \right) \cdot {\left( {1 - \frac{\theta }{{2\pi }}{P_t}} \right)^{K - 1}}\\
    \cdot \left\{ {1 + {{\left[ {1 - \frac{\theta }{{2\pi }}\left( {1 - {P_t}} \right) \cdot {{\left( {1 - \frac{\theta }{{2\pi }}{P_t}} \right)}^{K - 1}}} \right]}^{K - 1 - D\left( {t - 1} \right)}}} \right\}
    \end{array}
    \end{equation}
    with CRA and
    \begin{equation}
    \begin{array}{l}
    P_{{\rm{A}} \to {\rm{B}}}^{SBA}\left( t \right) = {P_t} \cdot \left( {1 - {P_t}} \right) \cdot {\left( {1 - {P_t}} \right)^{K - 1}}\\
    \cdot \left\{ {1 + {{\left[ {1 - \left( {1 - {P_t}} \right) \cdot {{\left( {1 - {P_t}} \right)}^{K - 1}}} \right]}^{K - 1 - D\left( {t - 1} \right)}}} \right\}
    \end{array}
    \end{equation}
     with SBA, where ${P_t}$ is the transmit probability and $D(t-1)$ is the number of neighbors that have
     discovered node A in the past $t - 1$ time slots in the beam where node B is located.
    \end{theorem}

    \begin{proof}
    As the ND process of a node in each beam is independent and approximately the same,
    the discovery probability of a neighboring node in a beam within a time slot represents
    the ability of the ND algorithm to find the neighbors in all beams.
    In the two-way ND algorithm with directional antennas,
    node A discovers its unknown neighbor B,
    which consists of two situations, as shown in  Fig. \ref{5aa} and Fig. \ref{5bb}.

\begin{figure}[!t]
\centering
\subfigure[Case 1]{
\includegraphics[width=3.5cm]{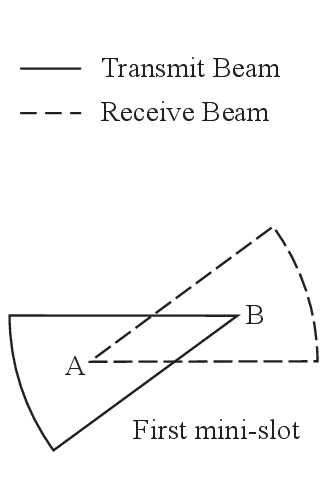}
\label{5aa}
}
\subfigure[Case 2]{
\includegraphics[width=3.5cm]{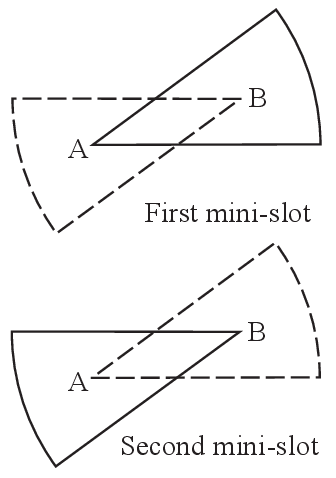}
\label{5bb}
}
\caption{Two cases of node A discovering its neighbor B.}
\end{figure}

    1) Case 1: In the first mini-slot, if node A is in receiving state,
    node B is in transmitting state,
    and other neighbors of node A do not interfere with the reception of node A,
    node A receives the signal of node B with probability
    \begin{equation}
    P_R^{CRA}{\rm{ = }}\frac{\theta }{{2\pi }}\left( {1 - {P_t}} \right) \cdot \frac{\theta }{{2\pi }}{P_t} \cdot {\left( {1 - \frac{\theta }{{2\pi }}{P_t}} \right)^{K - 1}}
    \end{equation}
    using CRA and with probability
    \begin{equation}
    P_R^{SBA}{\rm{ = }}\left( {1 - {P_t}} \right) \cdot {P_t} \cdot {\left( {1 - {P_t}} \right)^{K - 1}}
    \end{equation}
    using SBA.

    2) Case 2: In the first mini-slot, if node A is in transmitting state,
    node B is in receiving state, and other neighbors of node B do not interfere with the reception of node B,
    node B successfully receives the signal of node A with probability
    \begin{equation}
    P_{T1}^{CRA}{\rm{ = }}\frac{\theta }{{2\pi }}{P_t} \cdot \frac{\theta }{{2\pi }}\left( {1 - {P_t}} \right) \cdot {\left( {1 - \frac{\theta }{{2\pi }}{P_t}} \right)^{K - 1}}
    \end{equation}
    using CRA and with probability
    \begin{equation}
    P_{T1}^{SBA} = {P_t} \cdot \left( {1 - {P_t}} \right) \cdot {\left( {1 - {P_t}} \right)^{K - 1}}
    \end{equation}
    using SBA.

    In the second mini-slot, except of node B,
    the $K-1-D(t-1)$ neighbors of node A that have not been discovered in the first $t-1$ time slots
    will reply to node A with probability
    \begin{equation}
    P_{reply}^{CRA} = \frac{\theta }{{2\pi }}\left( {1 - {P_t}} \right) \cdot {\left( {1 - \frac{\theta }{{2\pi }}{P_t}} \right)^{K - 1}}
    \end{equation}
    using CRA and with probability
    \begin{equation}
    P_{reply}^{SBA} = \left( {1 - {P_t}} \right) \cdot {\left( {1 - {P_t}} \right)^{K - 1}}
    \end{equation}
    using SBA.

    If these $K-1-D(t-1)$ neighbors do not reply to node A,
    node A can successfully receive the reply from node B with probability
    \begin{equation}
    P_{T2}^{CRA/SBA}\left( t \right) = {\left[ {1 - P_{reply}^{CRA/SBA}} \right]^{K - 1 - D\left( {t - 1} \right)}},
    \end{equation}
    where $P_x^{CRA/SBA} \in \{ P_x^{CRA},\;P_x^{SBA}\} $ and $x \in \{ T2,\;reply,\;{\rm{A}} \to {\rm{B}},\;R,\;T1\} $.

    Through the above analysis, node A discovers node B with probability
    \begin{equation}
    P_{{\rm{A}} \to {\rm{B}}}^{CRA/SBA}\left( t \right) = P_R^{CRA/SBA} + P_{T1}^{CRA/SBA}P_{T2}^{CRA/SBA}\left( t \right).
    \end{equation}
    \end{proof}

    \subsection{SIC based ND algorithms}
        \textbf{\textit{1) Perfect SIC}}
        \begin{lemma}
        \label{lemma_2}
        The maximum number of simultaneous unpacking data packets at the receiver
        by perfect SIC is (proof can be found in Appendix B)
        \begin{equation}
        \label{eq_17}
        {n_0} = \left\lfloor {2 + {{\log }_{1 + \beta }}\frac{{16{\pi ^2}{r^2}}}{{\lambda _0^2\beta }}} \right\rfloor ,
        \end{equation}
        where $\beta$ is the signal-to-interference ratio (SIR) threshold for successful unpacking and ${\lambda _0}$ is free space wavelength.
        \end{lemma}

        According to Lemma \ref{lemma_2}, ${n_0}$ is jointly determined by the unpacking SIR threshold $\beta $ and the communication radius $r$.
        When the number of collided data packets $M$ is greater than the maximum number of simultaneous unpacking data packets ${n_0}$ at the receiver by SIC,
        the probability that data packets are unpacked by SIC successfully is extremely low,
        which is ignored in the subsequent analysis.

        \begin{lemma}
        Suppose that the number of collided data packets at the receiver is $M$ ($M \le {n_0}$),
        and the distances between their transmitting nodes and the receiving node are ${d_1},\;{d_2},\; \ldots ,\;{d_M}$ from small to large.
        The expected probability that one of the $M$ collided packets can be successfully unpacked is (proof can be found in Appendix C)
        \begin{equation}
        \label{eq_L3}
        \bar P(M) = \frac{1}{M}\sum\limits_{i = 1}^M {\prod\limits_{j = 1}^i {\frac{i}{{\beta {r^2}\sum\limits_{k = 0}^{M - j - 1} {{{\left( {\frac{1}{{{d_{M - k}}}}} \right)}^2}} }}} } .
        \end{equation}
        \end{lemma}

        \begin{theorem}
        With CRA-SIC and SBA-SIC, the probability of the node successfully discovering a neighbor at $t$-th time slot is (\ref{eq_31}).
        \end{theorem}

        \begin{proof}
        Similar to the proof of Theorem 1, in the $t$-th time slot,
        there are two cases that node A can find its undiscovered neighbor B.

        1) Case 1: In the first mini-slot, node A is in receiving state,
        node B is in transmitting state, and node A is able to unpack the signal of node B from the received signals using SIC.
        Then node A can receive the signal of node B with probability in (\ref{eq_24})
        using CRA-SIC and with probability in (\ref{eq_25}) using SBA-SIC.

        2) Case 2: In the first mini-slot, node A is in transmitting state,
        node B is in receiving state, and node B is able to unpack the signal of node A from the received signals using SIC.
        Then, node B can successfully receive the signal of node A with probability in (\ref{eq_26})
        using CRA-SIC and with probability in (\ref{eq_27}) using SBA-SIC.

        In the second mini-slot, except node B, the $K-1$ neighbors of node A will reply to the direction of node A,
        if they received a signal from an undiscovered neighbor in the direction of node A in the first mini-slot.
        The nodes can be further classified into the nodes that have discovered node A discovering
        a new unknown node and nodes that have not discovered node A discovering node A
        with probability in (\ref{eq_28}) using CRA-SIC and
        with probability in (\ref{eq_29}) using SBA-SIC.

        If node A can unpack the signal of node B from the received reply signals,
        node A can successfully receive the reply from node B with probability in (\ref{eq_30}).

        Therefore, node A discovers node B with probability shown in (\ref{eq_31}).
        \end{proof}

\begin{figure*}[!t]
\normalsize
\begin{equation}
\label{eq_24}
P_R^{CRA\_SIC}{\rm{ = }}\frac{\theta }{{2\pi }}\left( {1 - {P_t}} \right) \cdot \frac{\theta }{{2\pi }}{P_t} \cdot \sum\limits_{m = 0}^{\min \left( {K - 1,{n_0} - 1} \right)} {C_{K - 1}^m{{\left( {\frac{\theta }{{2\pi }}{P_t}} \right)}^m}{{\left( {1 - \frac{\theta }{{2\pi }}{P_t}} \right)}^{K - 1 - m}}} \bar P\left( {m + 1} \right)
\end{equation}
\begin{equation}
\label{eq_25}
P_R^{SBA\_SIC}{\rm{ = }}\left( {1 - {P_t}} \right) \cdot {P_t} \cdot \sum\limits_{m = 0}^{\min \left( {K - 1,{n_0} - 1} \right)} {C_{K - 1}^mP_t^m{{\left( {1 - {P_t}} \right)}^{K - 1 - m}}} \bar P\left( {m + 1} \right)
\end{equation}
\begin{equation}
\label{eq_26}
P_{T1}^{CRA\_SIC}{\rm{ = }}\frac{\theta }{{2\pi }}{P_t} \cdot \frac{\theta }{{2\pi }}\left( {1 - {P_t}} \right) \cdot \sum\limits_{m = 0}^{\min \left( {K - 1,{n_0} - 1} \right)} {C_{K - 1}^m{{\left( {\frac{\theta }{{2\pi }}{P_t}} \right)}^m}{{\left( {1 - \frac{\theta }{{2\pi }}{P_t}} \right)}^{K - 1 - m}}} \bar P\left( {m + 1} \right)
\end{equation}
\begin{equation}
\label{eq_27}
P_{T1}^{SBA\_SIC}{\rm{ = }}{P_t} \cdot \left( {1 - {P_t}} \right) \cdot \sum\limits_{m = 0}^{\min \left( {K - 1,{n_0} - 1} \right)} {C_{K - 1}^mP_t^m{{\left( {1 - {P_t}} \right)}^{K - 1 - m}}} \bar P\left( {m + 1} \right)
\end{equation}
\begin{equation}
\label{eq_28}
\begin{array}{l}
P_{reply}^{CRA\_SIC} = \frac{{D\left( {t - 1} \right)}}{K} \cdot \frac{\theta }{{2\pi }}\left( {1 - {P_t}} \right) \cdot C_{K - D\left( {t - 1} \right)}^1\frac{\theta }{{2\pi }}{P_t} \cdot \sum\limits_{n = 0}^{\min \left( {K - 2,{n_0} - 2} \right)} {C_{K - 2}^n{{\left( {\frac{\theta }{{2\pi }}{P_t}} \right)}^n}{{\left( {1 - \frac{\theta }{{2\pi }}{P_t}} \right)}^{K - 2 - n}}} \bar P\left( {n + 2} \right)\\
\;\;\;\;\;\;\;\;\;\;\;\;\;\;\; + \frac{{K - D\left( {t - 1} \right)}}{K} \cdot \frac{\theta }{{2\pi }}\left( {1 - {P_t}} \right) \cdot \sum\limits_{n = 0}^{\min \left( {K - 1,{n_0} - 1} \right)} {C_{K - 1}^n{{\left( {\frac{\theta }{{2\pi }}{P_t}} \right)}^n}{{\left( {1 - \frac{\theta }{{2\pi }}{P_t}} \right)}^{K - 1 - n}}} \bar P\left( {n + 1} \right)
\end{array}
\end{equation}
\begin{equation}
\label{eq_29}
\begin{array}{l}
P_{reply}^{SBA\_SIC} = \frac{{D\left( {t - 1} \right)}}{K} \cdot \left( {1 - {P_t}} \right) \cdot C_{K - D\left( {t - 1} \right)}^1{P_t} \cdot \sum\limits_{n = 0}^{\min \left( {K - 2,{n_0} - 2} \right)} {C_{K - 2}^nP_t^n{{\left( {1 - {P_t}} \right)}^{K - 2 - n}}} \bar P\left( {n + 2} \right)\\
\;\;\;\;\;\;\;\;\;\;\;\;\;\;\; + \frac{{K - D\left( {t - 1} \right)}}{K} \cdot \left( {1 - {P_t}} \right) \cdot \sum\limits_{n = 0}^{\min \left( {K - 1,{n_0} - 1} \right)} {C_{K - 1}^nP_t^n{{\left( {1 - {P_t}} \right)}^{K - 1 - n}}} \bar P\left( {n + 1} \right)
\end{array}
\end{equation}
\begin{equation}
\label{eq_30}
P_{T2}^{CRA/SBA\_SIC}\left( t \right) = \sum\limits_{m = 0}^{\min \left( {K - 1,{n_0} - 1} \right)} {C_{K - 1}^m} {\left( {P_{reply}^{CRA/SBA\_SIC}} \right)^m}{\left( {{\rm{1 - }}P_{reply}^{CRA/SBA\_SIC}} \right)^{K - 1 - m}}\bar P\left( {m + 1} \right)
\end{equation}
\begin{equation}
\label{eq_31}
P_{{\rm{A}} \to {\rm{B}}}^{CRA/SBA\_SIC}\left( t \right) = P_R^{CRA/SBA\_SIC} + P_{T1}^{CRA/SBA\_SIC}P_{T2}^{CRA/SBA\_SIC}\left( t \right)
\end{equation}
\hrulefill \vspace*{4pt}
\end{figure*}

        \textbf{\textit{2) Imperfect SIC}}

        \begin{lemma}
        The maximum number of simultaneous unpacking data packets by imperfect SIC
        is upper bounded by $\left\lfloor {2 + {{\log }_{1 + \beta }}\frac{{16{\pi ^2}{r^2}}}{{\lambda _0^2\beta }}} \right\rfloor $.
        In Appendix D we provide a proof.
        \end{lemma}

    \subsection{SIC and MPR based ND algorithms}

    \begin{theorem}
    With CRA-SIC-MPR and SBA-SIC-MPR,
    the probability of a node discovering a neighbor at $t$-th time slot is (\ref{eq_44}).
    \end{theorem}

    \begin{proof}
    Suppose that there are $h$ different modulation methods in the network.
    Each node randomly selects a modulation method with probability $\frac{1}{h}$ to modulate the packet to be sent.

    Similar to the proof of Theorem 1,
    in the $t$-th time slot,
    the process of node A discovering its undiscovered neighbor B can be divided into the following two cases.

    1) Case 1: In the first mini-slot, node A is in receiving state, node B is in transmitting state,
    and node A is able to unpack the signal of node B from the received signals, i.e.,
    the signal of node B has a different modulation method compared with other signals or can be unpacked by the SIC.
    Thus, node A can successfully receive the signal from node B with probability in (\ref{eq_37})
    using CRA-SIC-MPR and with probability in (\ref{eq_38}) using SBA-SIC-MPR.

    2) Case 2: In the first mini-slot, node A is in transmitting state, node B is in receiving state,
    and node B can unpack the signal of node A from the received signals.
    Then, node B can successfully receive the signal of node A with probability in (\ref{eq_39})
    using CRA-SIC-MPR and with probability in (\ref{eq_40}) using SBA-SIC-MPR.

    In the second mini-slot, except node B,
    the $K-1$ neighbors of node A will reply to the direction of node A if they received the signals
    from an undiscovered neighbor in the direction of node A in the first mini-slot.
    The nodes can be further classified into the nodes that have discovered node A discovering
    a new node and nodes that have not discovered node A discovering node A,
    whose probability is (\ref{eq_41})
    using CRA-SIC-MPR and (\ref{eq_42}) using SBA-SIC-MPR.

    If node A can unpack the signal of node B from the received reply signals,
    node A can successfully receive the reply from
    node B with the probability in (\ref{eq_43}).

    Therefore, node A discovers node B with the probability shown in (\ref{eq_44}).
    \end{proof}

    Fig. \ref{a1} compares the variation of $P_{{\rm{A}} \to {\rm{B}}}^x\left( t \right)$ of
    the above six ND algorithms as the number of discovered neighbors increases.
    It is revealed that the introduction of SIC and MPR always increases the discovery probability.
    More analysis and discussion are presented in Section VI.

\begin{figure}[!t]
\centering
\includegraphics[width=0.5\textwidth]{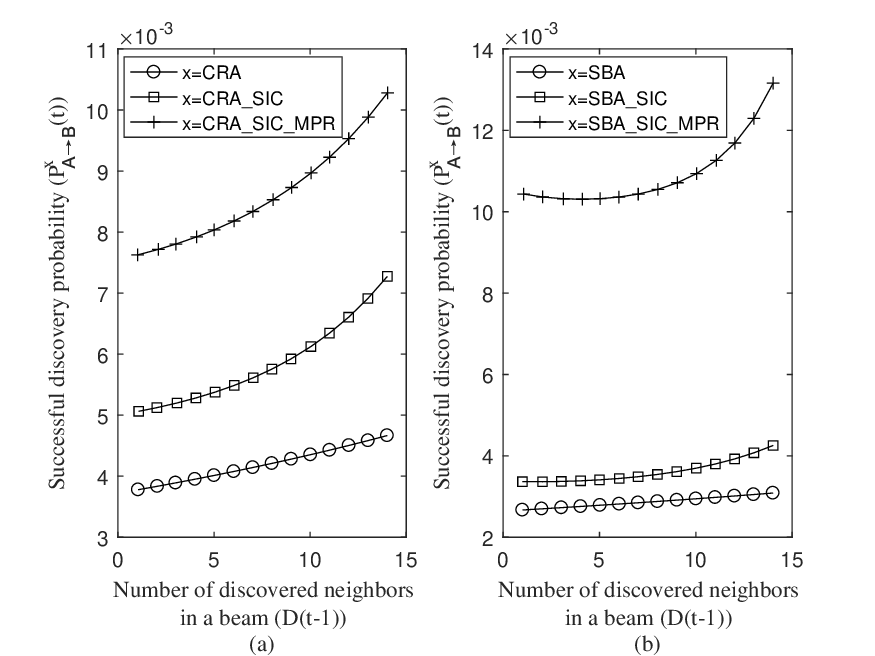}
\caption{Discover probability for varying number of discovered neighbors ($\theta  = \frac{\pi }{2}$, ${P_t} = 0.15$, $h=2$, $K=15$)}
\label{a1}
\end{figure}

\begin{figure*}[!t]
\normalsize
\begin{equation}
\label{eq_37}
P_R^{CRA\_SIC\_MPR}{\rm{ = }}\frac{\theta }{{2\pi }}\left( {1 - {P_t}} \right) \cdot \frac{\theta }{{2\pi }}{P_t} \cdot \sum\limits_{m = 0}^{\min \left( {K - 1,{n_0} - 1} \right)} {C_{K - 1}^m{{\left( {\frac{\theta }{{2\pi h}}{P_t}} \right)}^m}{{\left( {1 - \frac{\theta }{{2\pi h}}{P_t}} \right)}^{K - 1 - m}}}\bar P\left( {m + 1} \right)
\end{equation}
\begin{equation}
\label{eq_38}
P_R^{SBA\_SIC\_MPR}{\rm{ = }}\left( {1 - {P_t}} \right) \cdot {P_t} \cdot \sum\limits_{m = 0}^{\min \left( {K - 1,{n_0} - 1} \right)} {C_{K - 1}^m{{\left( {\frac{{{P_t}}}{h}} \right)}^m}{{\left( {1 - \frac{{{P_t}}}{h}} \right)}^{K - 1 - m}}}\bar P\left( {m + 1} \right)
\end{equation}
\begin{equation}
\label{eq_39}
P_{T1}^{CRA\_SIC\_MPR}{\rm{ = }}\frac{\theta }{{2\pi }}{P_t} \cdot \frac{\theta }{{2\pi }}\left( {1 - {P_t}} \right) \cdot \sum\limits_{m = 0}^{\min \left( {K - 1,{n_0} - 1} \right)} {C_{K - 1}^m{{\left( {\frac{\theta }{{2\pi h}}{P_t}} \right)}^m}{{\left( {1 - \frac{\theta }{{2\pi h}}{P_t}} \right)}^{K - 1 - m}}}\bar P\left( {m + 1} \right)
\end{equation}
\begin{equation}
\label{eq_40}
P_{T1}^{SBA\_SIC\_MPR}{\rm{ = }}{P_t} \cdot \left( {1 - {P_t}} \right) \cdot \sum\limits_{m = 0}^{\min \left( {K - 1,{n_0} - 1} \right)} {C_{K - 1}^m{{\left( {\frac{{{P_t}}}{h}} \right)}^m}{{\left( {1 - \frac{{{P_t}}}{h}} \right)}^{K - 1 - m}}}\bar P\left( {m + 1} \right)
\end{equation}
\begin{equation}
\label{eq_41}
\begin{array}{l}
P_{reply}^{CRA\_SIC\_MPR} = \frac{{D\left( {t - 1} \right)}}{K} \cdot \frac{\theta }{{2\pi }}\left( {1 - {P_t}} \right) \cdot C_{K - D\left( {t - 1} \right)}^1\frac{\theta }{{2\pi }}{P_t}\\
\;\;\;\;\;\;\;\;\;\;\;\;\;\;\;\;\;\;\;\;\;\;\;\;\;\;\; \cdot \left[ {\frac{1}{h} \cdot \sum\limits_{n = 0}^{\min \left( {K - 2,{n_0} - 2} \right)} {C_{K - 2}^n{{\left( {\frac{\theta }{{2\pi h}}{P_t}} \right)}^n}{{\left( {1 - \frac{\theta }{{2\pi h}}{P_t}} \right)}^{K - 2 - n}}} \bar P\left( {n + 2} \right)} \right.\\
\;\;\;\;\;\;\;\;\;\;\;\;\;\;\;\;\;\;\;\;\;\;\;\;\;\;\left. { + \left( {1 - \frac{1}{h}} \right) \cdot \sum\limits_{n = 0}^{\min \left( {K - 2,{n_0} - 1} \right)} {C_{K - 2}^n{{\left( {\frac{\theta }{{2\pi h}}{P_t}} \right)}^n}{{\left( {1 - \frac{\theta }{{2\pi h}}{P_t}} \right)}^{K - 2 - n}}} \bar P\left( {n + 1} \right)} \right]\\
\;\;\;\;\;\;\;\;\;\;\;\;\;\;\;\;\;\;\;\;\;\;\;\;\;\; + \frac{{K - D\left( {t - 1} \right)}}{K} \cdot \frac{\theta }{{2\pi }}\left( {1 - {P_t}} \right) \cdot \sum\limits_{n = 0}^{\min \left( {K - 1,{n_0} - 1} \right)} {C_{K - 1}^n{{\left( {\frac{\theta }{{2\pi h}}{P_t}} \right)}^n}{{\left( {1 - \frac{\theta }{{2\pi h}}{P_t}} \right)}^{K - 1 - n}}} \bar P\left( {n + 1} \right)
\end{array}
\end{equation}
\begin{equation}
\label{eq_42}
\begin{array}{l}
P_{reply}^{SBA\_SIC\_MPR} = \frac{{D\left( {t - 1} \right)}}{K} \cdot \left( {1 - {P_t}} \right) \cdot C_{K - D\left( {t - 1} \right)}^1{P_t}\\
 \cdot \left[ {\frac{1}{h} \cdot \sum\limits_{n = 0}^{\min \left( {K - 2,{n_0} - 2} \right)} {C_{K - 2}^n{{\left( {\frac{{{P_t}}}{h}} \right)}^n}{{\left( {1 - \frac{{{P_t}}}{h}} \right)}^{K - 2 - n}}} \bar P\left( {n + 2} \right) + (1 - \frac{1}{h}) \cdot \sum\limits_{n = 0}^{\min \left( {K - 2,{n_0} - 1} \right)} {C_{K - 2}^n{{\left( {\frac{{{P_t}}}{h}} \right)}^n}{{\left( {1 - \frac{{{P_t}}}{h}} \right)}^{K - 2 - n}}} \bar P\left( {n + 1} \right)} \right]\\
 + \frac{{K - D\left( {t - 1} \right)}}{K} \cdot \left( {1 - {P_t}} \right) \cdot \sum\limits_{n = 0}^{\min \left( {K - 1,{n_0} - 1} \right)} {C_{K - 1}^n{{\left( {\frac{{{P_t}}}{h}} \right)}^n}{{\left( {1 - \frac{{{P_t}}}{h}} \right)}^{K - 1 - n}}} \bar P\left( {n + 1} \right)
\end{array}
\end{equation}
\begin{equation}
\label{eq_43}
P_{T2}^{CRA/SBA\_SIC\_MPR}\left( t \right) = \sum\limits_{m = 0}^{\min \left( {K - 1,{n_0} - 1} \right)} {C_{K - 1}^m} {\left( {\frac{1}{h} \cdot P_{reply}^{CRA/SBA\_SIC\_MPR}} \right)^m}{\left( {{\rm{1 - }}\frac{1}{h} \cdot P_{reply}^{CRA/SBA\_SIC\_MPR}} \right)^{K - 1 - m}}\bar P\left( {m + 1} \right)
\end{equation}
\begin{equation}
\label{eq_44}
P_{{\rm{A}} \to {\rm{B}}}^{CRA/SBA\_SIC\_MPR}\left( t \right) = P_R^{CRA/SBA\_SIC\_MPR} + P_{T1}^{CRA/SBA\_SIC\_MPR}P_{T2}^{CRA/SBA\_SIC\_MPR}\left( t \right)
\end{equation}
\hrulefill \vspace*{4pt}
\end{figure*}

    \subsection{Expected number of time slots}

    \begin{theorem}
    The expectation of the number of time slots expected by a node to find all neighbors in all beams is
    \begin{equation}
    E\left( {{T_{all}}} \right) = \frac{{2\pi }}{\theta }\sum\limits_{j = 0}^{K - 1} {\frac{1}{{\left( {K - j} \right){P_{{\rm{A}} \to {\rm{B}}}}\left( {t|D\left( {t - 1} \right) = j} \right)}}} .
    \end{equation}
    \end{theorem}

    \begin{proof}
    In a beam, when the number of neighbors that have been discovered by node A is $j$,
    the probability that node A finds the next neighbor is
    \begin{equation}
    {P_{next}} = \left( {K - j} \right){P_{{\rm{A}} \to {\rm{B}}}}\left( {t|D\left( {t - 1} \right) = j} \right).
    \end{equation}

    Furthermore, assuming that the number of time slots required by node A to
    discovery the next neighbor in the beam is ${T_j}$,
    the expectation of ${T_j}$ is
    \begin{equation}
    \label{eq_338}
    \begin{array}{*{20}{l}}
    {E\left( {{T_j}} \right){\rm{ = }}\sum\limits_{t = 1}^\infty  {{{\left( {1 - {P_{next}}} \right)}^{t - 1}} \cdot {P_{next}} \cdot {t}} }\\
    {\;\;\;\;\;\;\;\;\;\;\;{\rm{ = }}\frac{1}{{{P_{next}}}}}\\
    {\;\;\;\;\;\;\;\;\;\;\; = \frac{1}{{\left( {K - j} \right){P_{suc}}(t|D\left( {t - 1} \right) = j)}},}
    \end{array}
    \end{equation}
    where $j$ is the number of neighbors that have been discovered in the beam.

    Summing all $E\left( {{T_j}} \right)$ from $j=0$ to $j=K-1$,
    the expectation of the number of time slots required for node A to find all neighbors in a beam can be obtained.
    Then, multiplying by $\frac{{2\pi }}{\theta }$ yields the number of time slots expected by node A to find all neighbors in all beams,
    which is obtained as Theorem 4.
    \end{proof}

    The proposed ND algorithms are designed to reduce the ND time.
    Fig. \ref{a2} compares the number of time slots required to discover 95\% of neighbors of the above six algorithms with different numbers of nodes.
    It is revealed from the figure that the introduction of both SIC and MPR reduces the ND time.
    More analysis and discussion are presented in Section VI.

\begin{figure}[!t]
\centering
\includegraphics[width=0.5\textwidth]{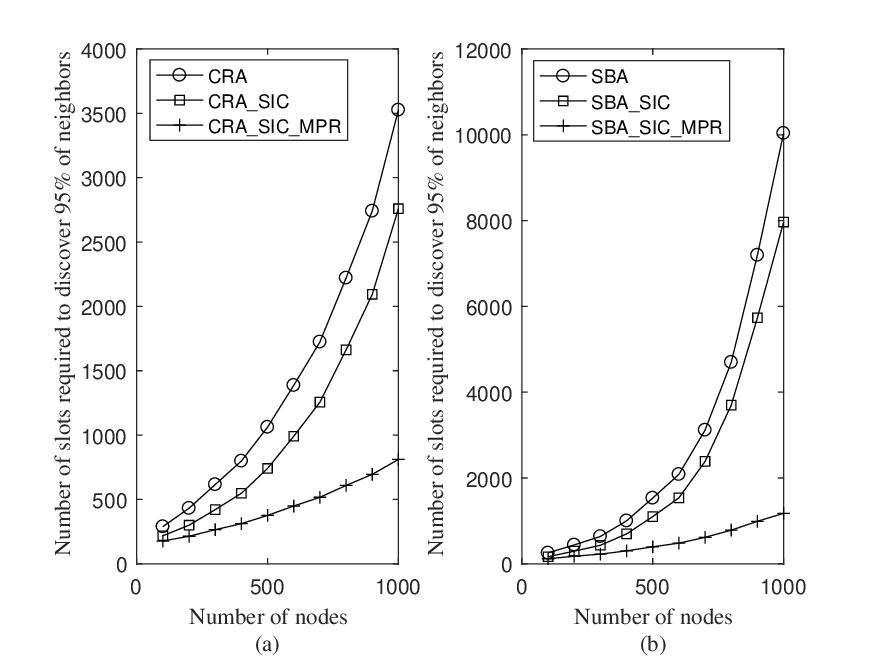}
\caption{Number of slots required to discover 95\% of neighbors for varying number of nodes ($h=2$)}
\label{a2}
\end{figure}

\section{Simulation Results and Analysis}
    In this section, the comparison between analytical results and simulation results,
    the effects of different parameters on ND time and the performance improvement of
    ND algorithms based on SIC and MPR are simulated.
    This section assumes that the nodes are randomly and uniformly placed in
    an area of $3000\;{\rm{m}} \times 3000\;{\rm{m}}$ and the transmission range is 800 m.

    \subsection{Verification of theoretical derivation}
    In Section V, we derive the discovery probabilities and the expected number of time slots
    of the algorithms proposed in this paper.
    Fig. \ref{7ST}(a) and Fig. \ref{7ST}(b) show the theoretical and simulation results of
    the fraction of discovered neighbors versus the number of time slots
    with the algorithms based on CRA and the algorithms based on SBA respectively when
    the number of nodes $N$ is 300 and the number of modulation methods $h$ is 3.
    The trends and values of the theoretical and simulation results are very close,
    which proves the theoretical derivation is reasonable and correct.
    In the following, we analyze proposed ND algorithms from the perspective of simulation, and its correctness is supported by Fig. \ref{7ST}.

\begin{figure}[!t]
\centering
\includegraphics[width=0.5\textwidth]{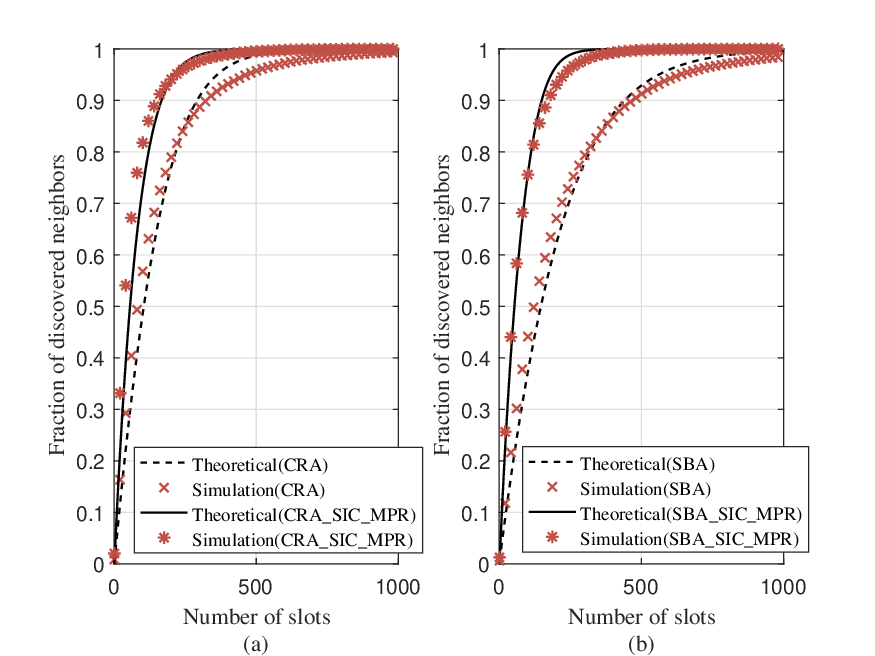}
\caption{Comparison of theoretical results and simulation results ($N=300$, $h=3$)}
\label{7ST}
\end{figure}

    \subsection{ND algorithms with and without SIC}
    Firstly, we compare SBA, CRA, SBA-SIC, and CRA-SIC with different numbers of nodes.
    Fig. \ref{7} illustrates the fraction of discovered neighbors,
    which is the ratio of the number of discovered neighbors to the number of total neighbors,
    during the ND processes of the above four algorithms when the unpacking threshold $\beta$ is 4,
    ${P_t}$ is 0.5, $\theta$ is $\frac{\pi }{6}$ and the numbers of nodes are 100, 250 and 500.

\begin{figure*}[!t]
\centering
\hspace{-7mm}
\subfigure[ $N=100$]{
\includegraphics[width=6cm]{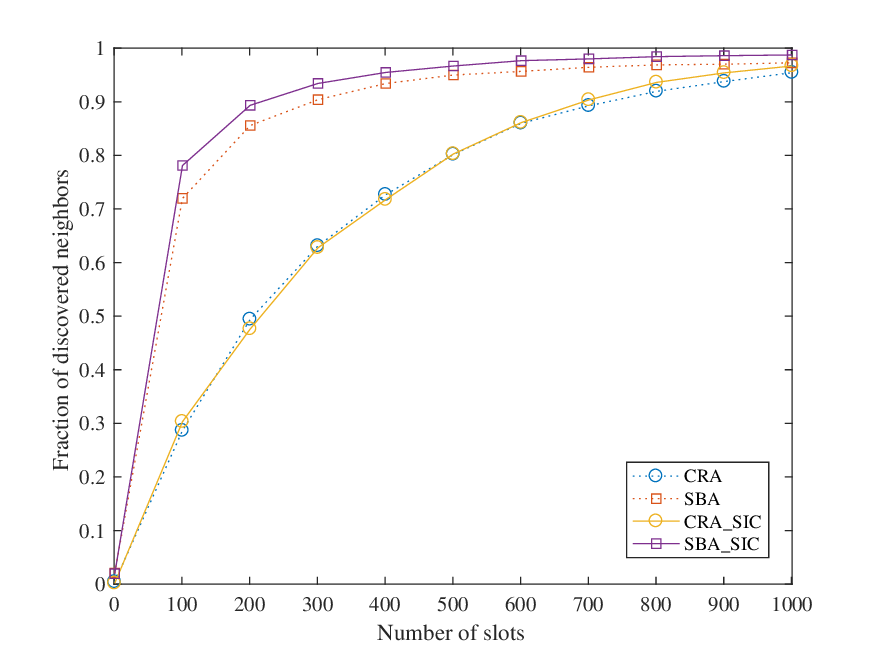}
\label{7a}
}
\hspace{-7mm}
\subfigure[ $N=250$]{
\includegraphics[width=6cm]{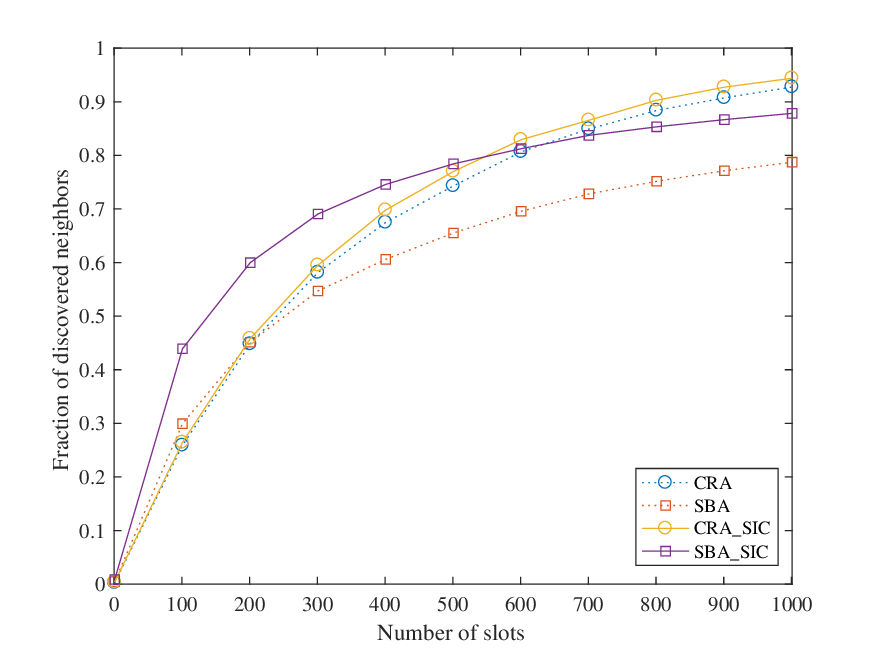}
\label{7b}
}
\hspace{-7mm}
\subfigure[ $N=500$]{
\includegraphics[width=6cm]{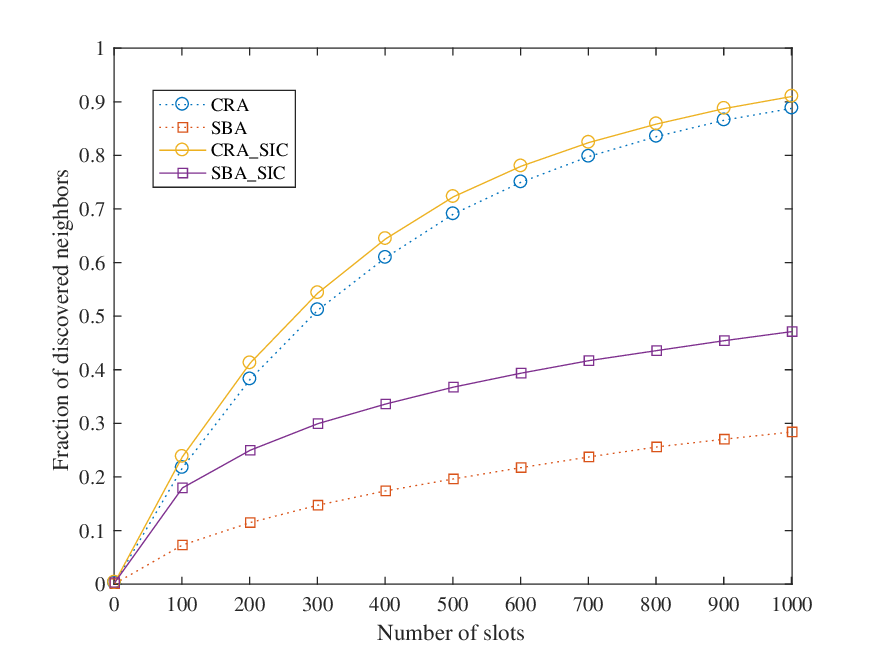}
\label{7c}
}
\hspace{-7mm}
\caption{Neighbor discovery process with different number of nodes (${P_t} = 0.5$, $\theta  = \frac{\pi }{6}$).}
\label{7}
\end{figure*}

    According to Fig. \ref{7}, the following conclusions are drawn.
    \begin{enumerate}
    \item The number of time slots required to complete ND is increasing
    with the increase of the number of nodes for the above four ND algorithms,
    since more nodes require more handshakes to establish network topology.
    \item When there are few nodes ($N=100$),
    SBA based ND algorithms have better performance compared with CRA based ND algorithms,
    since SBA based ND algorithms have a pre-defined scan order,
    the transmit and receive beams are easier to be aligned.
    However, when there are many nodes ($N=500$), the ND algorithms based on CRA have better performance,
    since the packets are not prone to be collided
    with the randomly selected beam characteristics of CRA based algorithms.
    When the number of nodes is 250,
    the SBA based ND algorithms are fast in the early stage and slow down
    when the number of time slots increases.
    \item The ND time of SIC based algorithms is smaller than that
    of the ND algorithms without SIC,
    since SIC is able to unpack collided packets.
    \end{enumerate}

        Then, the performance of CRA and SBA algorithms based on perfect SIC
        with respect to the beam width of the directional antenna $\theta $ and
        the probability of node transmission ${P_t}$ is revealed in Fig. \ref{8} and Fig. \ref{9}, respectively.

        Fig. \ref{8} illustrates the efficiency of SBA-SIC and CRA-SIC under different $\theta$ and ${P_t}$
        when $N$ is 300 and the number of time slots $T$ is 400.
        For SBA-SIC, when the ${P_t}$ is smaller than 0.1, the efficiency of ND algorithms improves with the increase of beam width $\theta$.
        When the ${P_t}$ is larger than 0.1,
        the ND algorithms with small beam width are more efficient than those with large beam width.
        The reason is that the number of neighbors in a beam increases when $\theta$ increases.
        When ${P_t}$ is large, a small $\theta$ reduces the number of neighbors in transmitting state,
        decreasing the probability of packet collision.
        Similarly, when ${P_t}$ is small, a large $\theta$ ensures fast ND.
        For SBA-SIC, the optimal ${P_t}$ is around 0.1.
        When $P_t=0.1$, the ND algorithms have similar performance.
        The phenomenon for CRA-SIC is similar to SBA-SIC.
        Since nodes randomly select beams, it is more difficult for them to find neighbors compared with SBA-SIC.
        Hence, the optimal ${P_t}$ and $\theta $ of CRA-SIC are both larger than those of SBA-SIC.
        Moreover, the optimal ${P_t}$ for CRA-SIC algorithms increases
        with the decrease of beam width, as shown in Fig. \ref{8b}.

\begin{figure}[!t]
\centering
\hspace{-10mm}
\subfigure[SBA-SIC]{
\includegraphics[width=5cm]{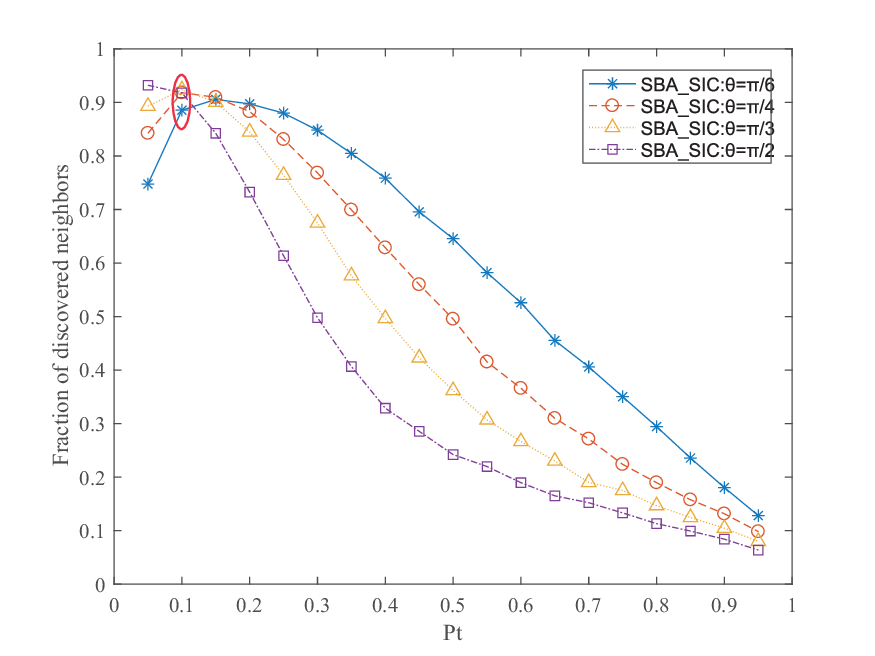}
\label{8a}
}
\hspace{-9mm}
\subfigure[CRA-SIC]{
\includegraphics[width=5cm]{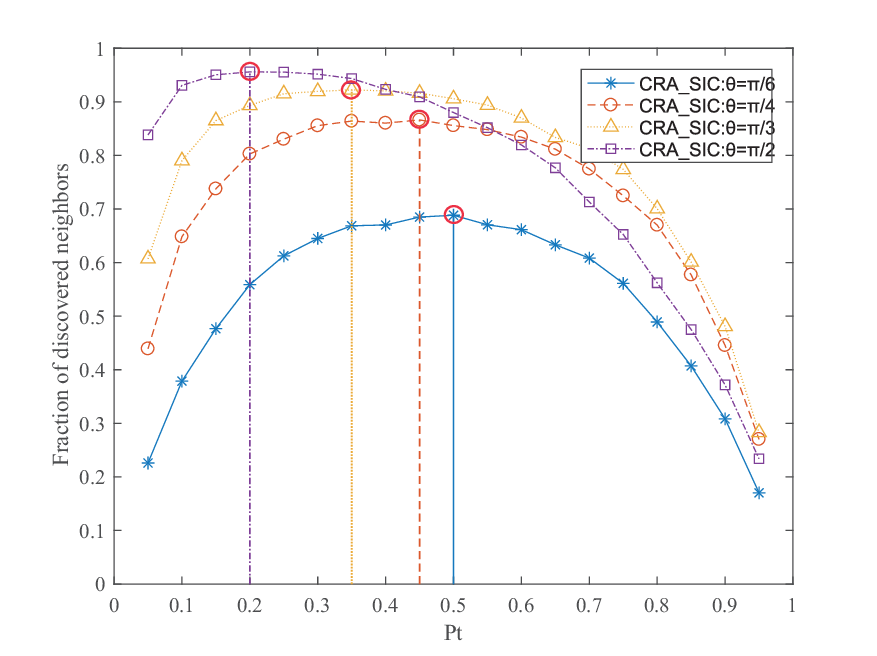}
\label{8b}
}
\hspace{-10mm}
\caption{The influence of the beam width $\theta$ and the transmit probability ${P_t}$ on the efficiency of SBA-SIC and CRA-SIC ($N=300$, $T=400$).}
\label{8}
\end{figure}

        Fig. \ref{9} illustrates the impact of $\theta $ on the efficiency of SBA-SIC and CRA-SIC when $T$ is 200.
        $P_t$ is 0.1 and 0.2 to guarantee a better performance of SBA-SIC and CRA-SIC, respectively.
        The number of nodes ranges from 50 to 500,
        with 50 as the interval.
        When there are few nodes, the ND time decreases with the increase of $\theta$.
        The reason is that it is a waste of time to scan the beam without neighbors.
        As the number of nodes increases, the speed of ND algorithms with a large $\theta $ decreases fast.
        When the number of nodes is large, there are too many neighbors and frequent collisions in a beam,
        which results in the decrease of ND time.
        Meanwhile, by comparing Fig. \ref{9a} and Fig. \ref{9b},
        it is found that $\theta $ has a larger influence on CRA-SIC than that on SBA-SIC.

\begin{figure}[!t]
\centering
\hspace{-10mm}
\subfigure[SBA-SIC with ${P_t} = 0.1$]{
\includegraphics[width=5cm]{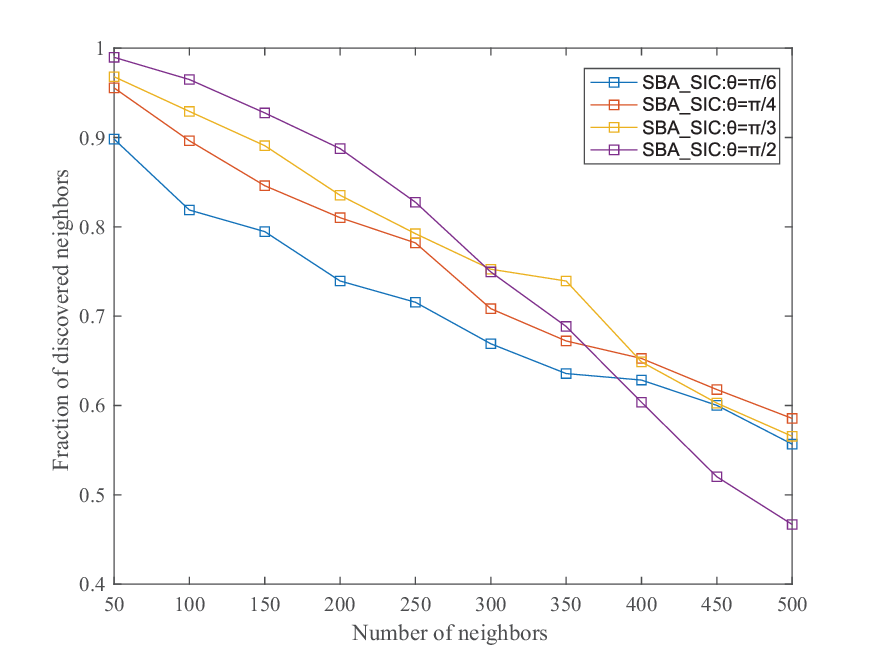}
\label{9a}
}
\hspace{-9mm}
\subfigure[CRA-SIC with ${P_t} = 0.2$]{
\includegraphics[width=5cm]{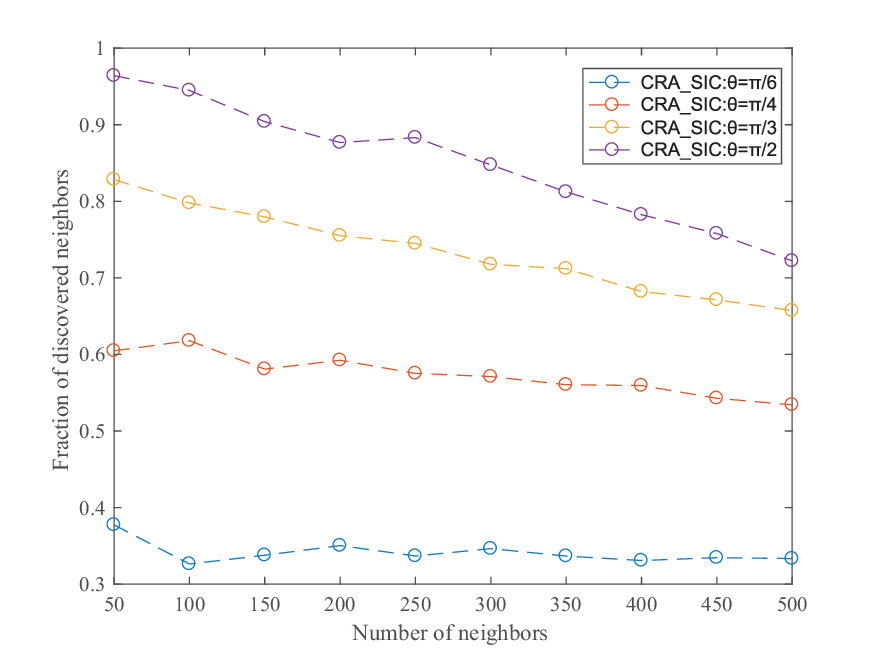}
\label{9b}
}
\hspace{-10mm}
\caption{The influence of the beam width $\theta$ on the efficiency of SBA-SIC and CRA-SIC ($T=200$).}
\label{9}
\end{figure}

        With obtaining the optimal parameters,
        we set the transmission probability of the ND algorithms based on SBA to 0.1
        and the beam width to $\frac{\pi }{3}$.
        For the ND algorithms based on CRA, we set the transmission probability to 0.2 and the beam width to $\frac{\pi }{2}$.
        Then, the ND algorithms with and without perfect SIC are compared and the results are shown in Fig. \ref{11}.
        The performance gap between CRA and SBA is not as obvious as the gap shown in Fig. \ref{7}.
        When introducing perfect SIC, the performance of both CRA and SBA has been improved,
        and the performance of SBA has been improved more obviously.
        Moreover, the performance improvement of the ND algorithms based on SIC is more significant in a network
        with denser nodes such as $N=500$ than that with $N=300$,
        since when $N=300$, there are fewer data packet collisions.
        However, there is an upper bound of the number of signals that can be unpacked by
        SIC simultaneously.
        When the density of nodes in the network exceeds a threshold,
        the role of SIC will be weakened.

\begin{figure}[!t]
\centering
\includegraphics[width=6.5cm]{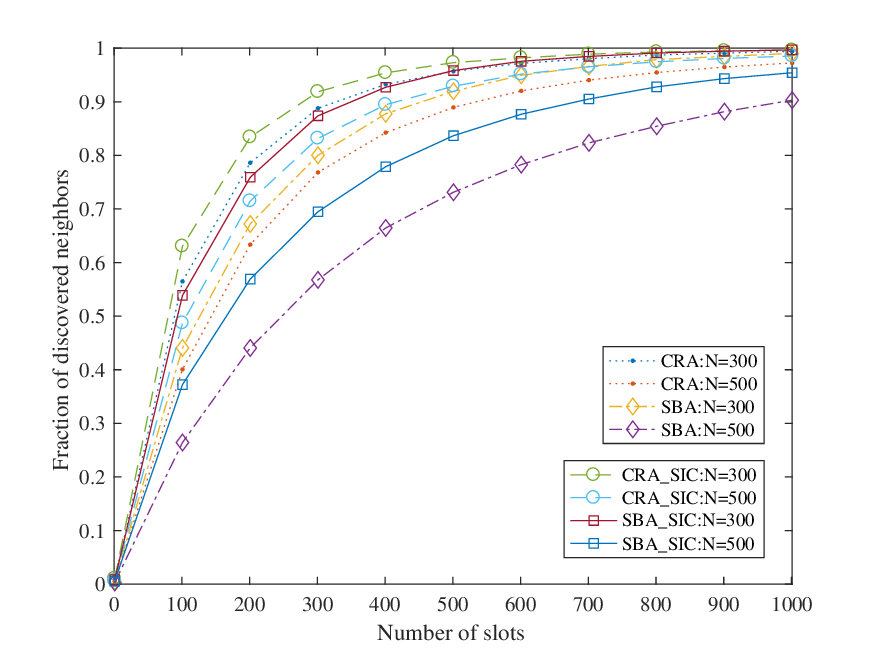}
\caption{Comparison of four ND algorithms with and without SIC using optimal parameters}
\label{11}
\end{figure}

        To study the influence of imperfect SIC on the performance of ND algorithms,
        the effect of imperfect SIC with different residual coefficient of interference cancellation $\xi$ on the performance improvement of ND algorithms is simulated when $N=800$, as shown in Fig. \ref{10}.
        It is revealed that the performance of CRA-SIC and SBA-SIC with imperfect SIC will be improved compared with CRA and SBA, respectively.
        When $\xi  < {10^{ - 2.5}}$, the imperfect SIC can be approximately regarded as the perfect SIC.
        Moreover, when ${10^{ - 2.5}} \le \xi  < {10^{ - 1}}$, the performance improvement of ND algorithms degrades rapidly.
        Tradeoffs can be made based on performance and implementation costs.
        When $\xi  \ge {10^{ - 1}}$, the performance improvement is not much different from the performance improvement when $\xi = 1$.
        The reason is that the power difference between the received signal with the highest power and other signals is relatively large.
        If $\xi$ is not small enough, the residual power of the received signal is too large,
        which affects the unpacking of other received signals.
        If the SIC applied by the ND algorithms cannot guarantee $\xi  < {10^{ - 1}}$,
        the requirements of SIC can be relaxed appropriately to improve the performance of ND algorithms at a small cost.

\begin{figure}[!t]
\centering
\includegraphics[width=6.5cm]{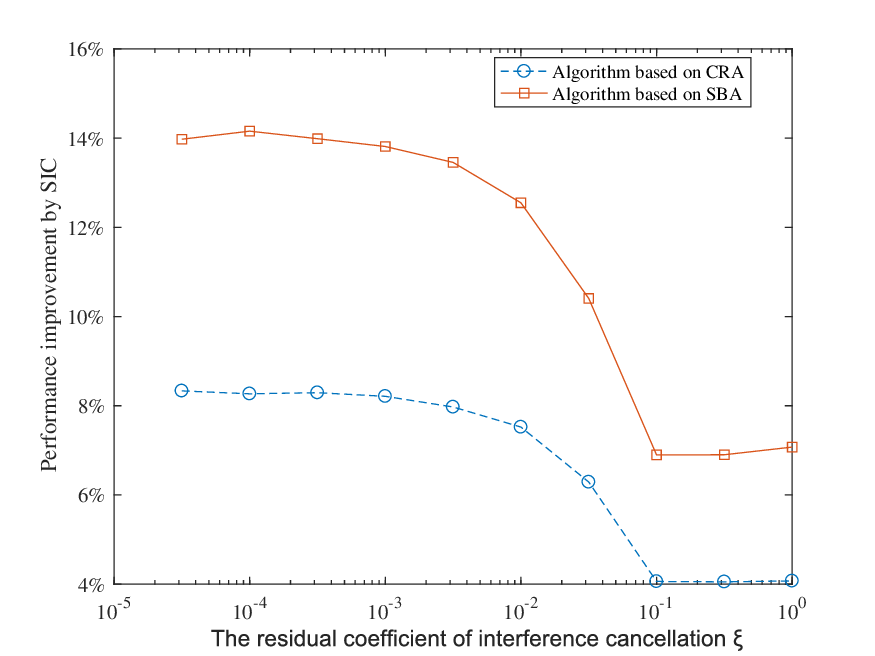}
\caption{The influence of the residual coefficient of interference cancellation $\xi$ on the efficiency of SBA-SIC and CRA-SIC ($N=800$).}
\label{10}
\end{figure}

    \subsection{ND algorithms with and without MPR}
    The four ND algorithms of SBA-SIC, CRA-SIC, SBA-SIC-MPR,
    and CRA-SIC-MPR with optimal parameters are compared in Fig. \ref{13}.
    The perfect SIC is applied and the number of modulation methods for MPR is set to 2.
    The introduction of MPR has improved the performance of both CRA and SBA.
    Besides, the performance improvement is more obvious in the networks with more nodes.
    This is due to the fact that the benefits of MPR are not obvious when there are few nodes.
    As the number of nodes and the amount of concurrent data increases,
    the benefits of MPR will be more significant.

\begin{figure}[!t]
\centering
\includegraphics[width=6.5cm]{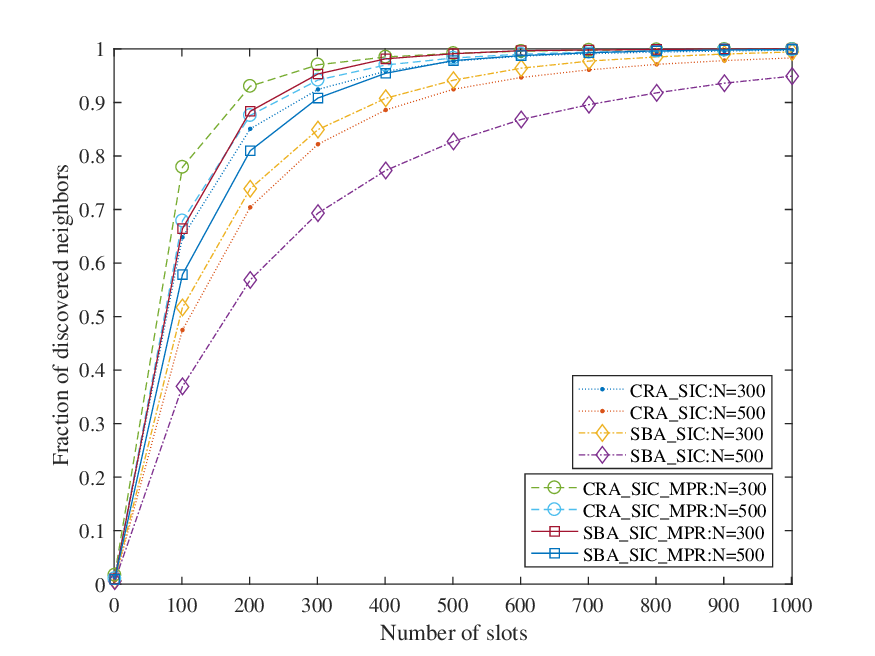}
\caption{Comparison of four ND algorithms with and without MPR}
\label{13}
\end{figure}

        The effect of different numbers of modulation methods $h$ on the performance improvement
        of the ND algorithms based on MPR under the optimal parameters with node size $N=300$ is simulated,
        as shown in Fig. \ref{14}.
        It is revealed that the introduction of MPR has a large improvement on the algorithm performance,
        and the improvement is more obvious with a larger number of modulation methods.

\begin{figure}[!t]
\centering
\hspace{-10mm}
\subfigure[SBA-SIC-MPR]{
\includegraphics[width=5cm]{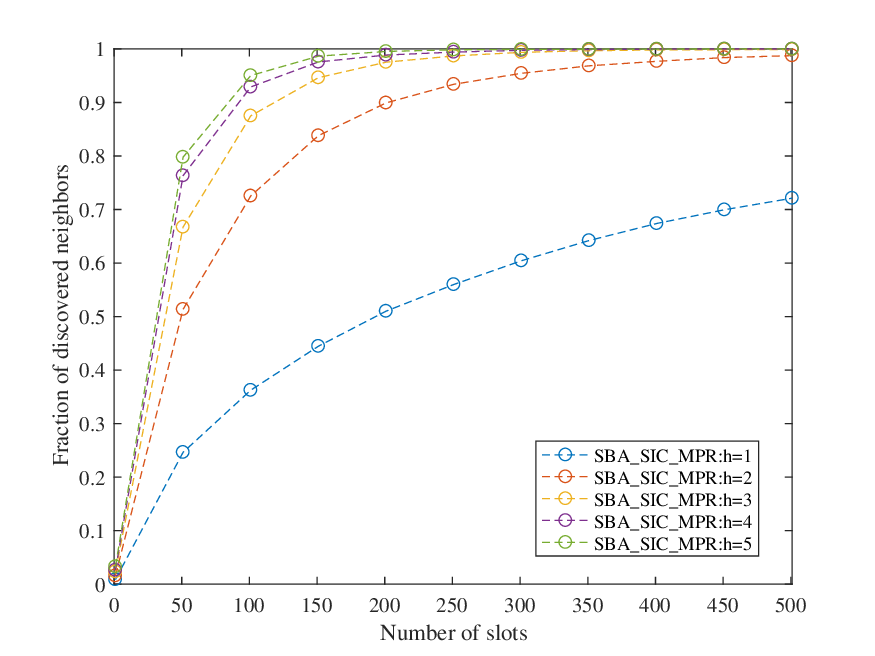}
\label{14a}
}
\hspace{-9mm}
\subfigure[CRA-SIC-MPR]{
\includegraphics[width=5cm]{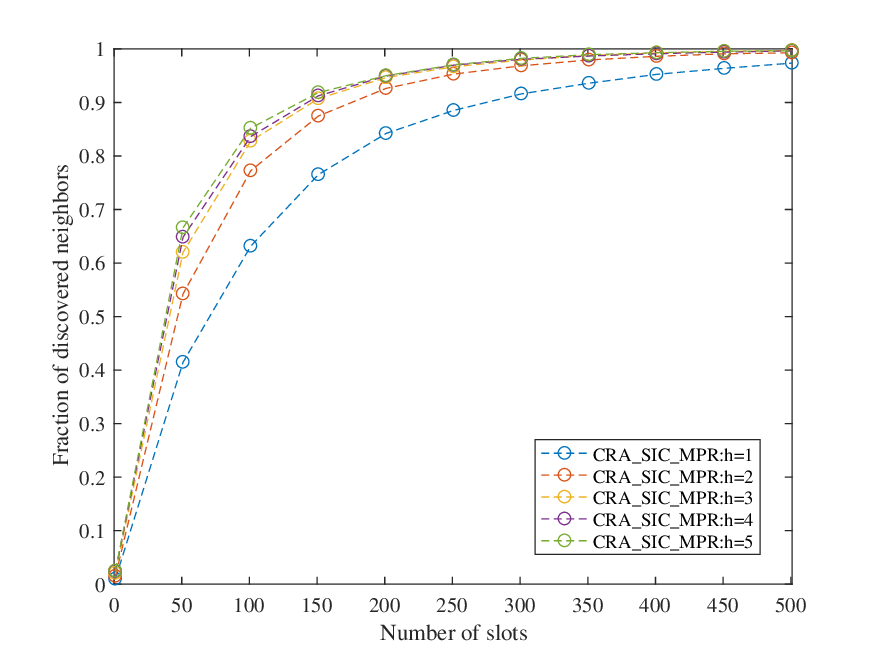}
\label{14b}
}
\hspace{-10mm}
\caption{The influence of the numbers of modulation methods $h$ on the efficiency of SBA-SIC-MPR and CRA-SIC-MPR ($N=300$)}
\label{14}
\end{figure}

    \subsection{Performance improvement}
    The proposed ND algorithms are designed to reduce the ND time.
    To intuitively compare the ND speed of the above six algorithms,
    Fig. \ref{a3} simulates the number of time slots required for these algorithms to discover 95\% of neighbors with different numbers of nodes.
    It is revealed from the figure that for both CRA and SBA,
    as the number of nodes in the network increases linearly,
    the ND time increases almost exponentially.
    When the number of nodes is large, it takes a long time to discover 95\% of neighbors.
    The analysis of Fig. \ref{a3} shows that the application of SIC reduces the ND time by 5\%-51\%,
    among which the ND time is shortened by 27.92\% and 26.88\% on average compared to SBA and CRA.
    In addition, on this basis, the application of MPR ($h=2$) further shortens the ND time to 23\%-93\%.
    Compared with the SBA and CRA, the average ND time is reduced by 69.02\% and 66.03\%.
    The above results are also consistent with the theoretical analysis in Fig. \ref{a2},
    which once again prove the correctness of the mathematical analysis.

\begin{figure}[!t]
\centering
\includegraphics[width=0.5\textwidth]{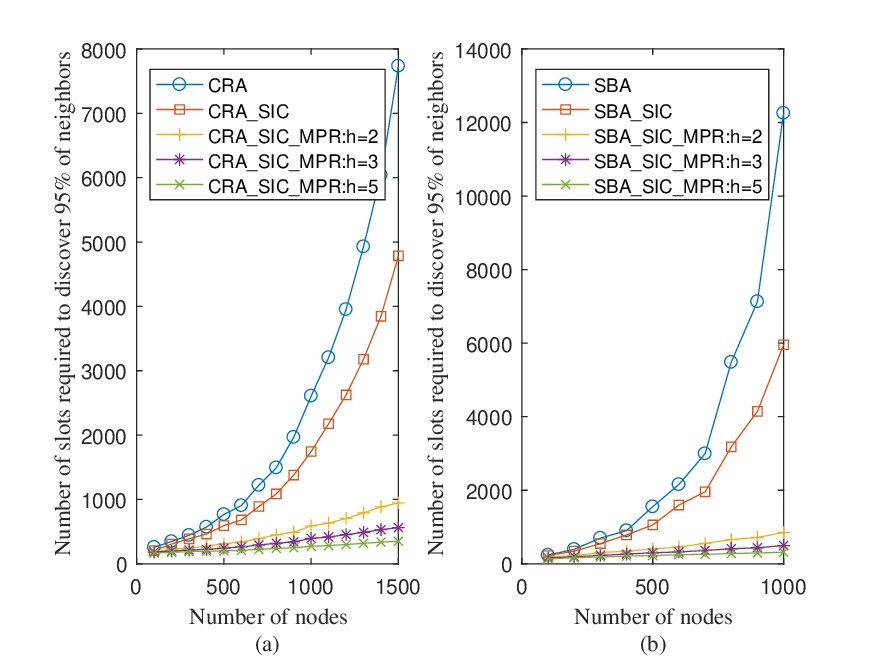}
\caption{Number of slots required to discover 95\% of neighbors for varying number of nodes}
\label{a3}
\end{figure}

\section{Conclusion}
In the IoT scenario, there is an urgent need to improve ND efficiency for the fast networking of massive nodes.
In this paper, SIC and MPR are introduced to shorten the ND time by solving the packet collision problem.
The time expectation of the ND algorithms is theoretically derived,
which verifies that the ND algorithms proposed in this paper achieve the expected ND in a lower time compared with traditional CRA and SBA.
The simulation shows that SIC declines the ND time of SBA by an average of 27.92\%,
while the ND time of CRA is reduced by an average of 26.88\% and the application of MPR and SIC declines the ND time of SBA by 69.02\% on average,
and CRA is reduced by 66.03\%.
In summary, both theoretical analysis and simulation results prove that the proposed scheme is more adaptable to the networking requirements in IoT scenarios.
This paper only focuses on the neighbor discovery and does not involve routing protocols.
Therefore, nodes' failures due to lack of energy are ignored.
In the future work, this aspect can be improved.

\section*{Appendix}
    \subsection{Proof of Lemma 1}
    The nodes are uniformly distributed in a rectangle with area $a*b$,
    which can be divided into two cases, as shown in Fig. \ref{5}.

\begin{figure}[!t]
\centering
\includegraphics[width=0.33\textwidth]{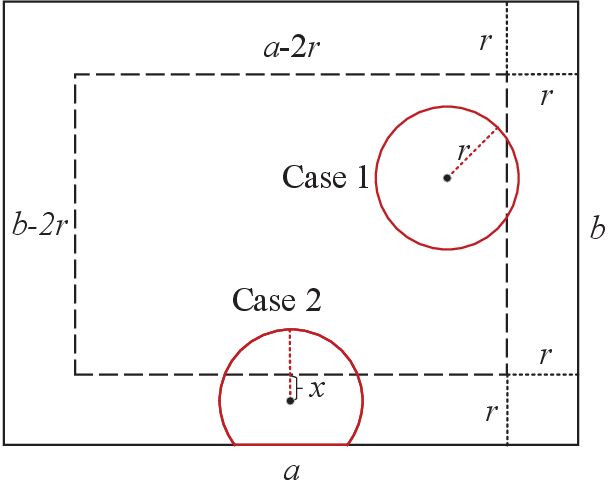}
\caption{Nodes distribution diagram}
\label{5}
\end{figure}

    1) When nodes are located in the small rectangle area of $\left( {a - 2r} \right)*\left( {b - 2r} \right)$,
    the average number of neighbors of nodes is
    \begin{equation}
    \label{eq_a_34}
    {N_{in}} = \lambda \pi {r^2},
    \end{equation}
    with probability
    \begin{equation}
    {P_{in}} = \frac{{\left( {a - 2r} \right)\left( {b - 2r} \right)}}{{ab}}.
    \end{equation}

    2) When nodes are located outside the small rectangle and the distance from it is $x$,
    the average number of neighbors of nodes is
    \begin{equation}
    {N_{out}} = \lambda \pi {r^2} - \lambda {r^2}\arccos \frac{{r - x}}{r} + \lambda \left( {r - x} \right)\sqrt {{r^2} - {{\left( {r - x} \right)}^2}} ,
    \end{equation}
    with probability
    \begin{equation}
    {P_{out}} = \frac{{2a + 2b - 8r + 8x}}{{ab}}.
    \end{equation}

    For the whole network, the average number of neighbors of a node is
    \begin{equation}
    \label{eq_a_38}
    \bar N = {N_{in}}{P_{in}} + \int_0^r {{N_{out}}{P_{out}}dx}.
    \end{equation}

    Combining equations(\ref{eq_a_34})-(\ref{eq_a_38}),
    we have equation (\ref{eq_L1}).
    \subsection{Proof of Lemma 2}
    Suppose that the number of collided data packets at the receiver is $M$,
    and the power of these $M$ data packets from strong to weak is ${S_1},\;{S_2},\; \ldots ,\;{S_M}$.
    Assume that the SIR threshold beyond which data packets can be unpacked is $\beta $.
    The conditions for SIC to successfully unpack all $M$ collision data packets are
        \begin{equation}
        \label{eq_12}
        \left\{ \begin{array}{l}
        SI{R_1} = \frac{{{S_1}}}{{{S_2} + {S_3} +  \ldots  + {S_M}}} \ge \beta ,\\
        SI{R_2} = \frac{{{S_2}}}{{{S_3} + {S_4} +  \ldots  + {S_M}}} \ge \beta ,\\
         \ldots \\
        SI{R_{M - 1}} = \frac{{{S_{M - 1}}}}{{{S_M}}} \ge \beta .
        \end{array} \right.
        \end{equation}
        In the free space transmission model,
        the relation between the power of received signal, i.e. $S$,
        and the distance between the receiving node and the transmitting node, i.e. $d$, is
        \begin{equation}
        \label{eq_13}
        S = {\left( {\frac{{{\lambda _0}}}{{4\pi d}}} \right)^2}{P_T}{G_T}{G_R}\;\;\;\;\frac{{{\lambda _0}}}{{4\pi }} \le d \le r,
        \end{equation}
        where ${P_T}$ is transmit power of nodes,
        ${G_T}$ is transmission gain of nodes and ${G_R}$ is reception gain of nodes.

        Substituting $d < \frac{{{\lambda _0}}}{{4\pi }}$ into (\ref{eq_13}),
        we have $S > {P_T}{G_T}{G_R}$, which is obviously not in line with the actual situation,
        because (\ref{eq_13}) only applies to the far-field,
        and $d < \frac{{{\lambda _0}}}{{4\pi }}$ is the condition of the near-field.
        The wireless ad hoc network generally communicates in the 2.4 GHz ISM frequency band,
        therefore $\frac{{{\lambda _0}}}{{4\pi }} \approx \frac{{12.5\;\rm{cm}}}{{4\pi }} \approx 1\;\rm{cm}$.
        Considering the distance between the two nodes is greater than 1 cm,
        we have (\ref{eq_14}) by substituting (\ref{eq_13}) into (\ref{eq_12}).
        \begin{equation}
        \label{eq_14}
        \left\{ {\begin{array}{*{20}{c}}
        {SI{R_1} = \frac{1}{{\sum\limits_{i = 0}^{M - 2} {{{\left( {\frac{{{d_1}}}{{{d_{M - i}}}}} \right)}^2}} }} \ge \beta ,}\\
        {SI{R_2} = \frac{1}{{\sum\limits_{i = 0}^{M - 3} {{{\left( {\frac{{{d_2}}}{{{d_{M - i}}}}} \right)}^2}} }} \ge \beta ,}\\
         \ldots \\
        {SI{R_{M - 1}} = \frac{{d_M^2}}{{d_{M - 1}^2}} \ge \beta ,}\\
        {\frac{{{\lambda _0}}}{{4\pi }} \le {d_1},\;{d_2},\; \ldots ,\;{d_M} \le r.}
        \end{array}} \right.
        \end{equation}

        According to (\ref{eq_14}), due to the limitations of the communication range and the unpacking threshold,
        the number of data packets that can be unpacked by SIC simultaneously is limited.
        To analyze the maximum number of data packets ${n_0}$ that can be unpacked by SIC at the same time,
        substituting ${d_M} = r$ and $SIR = \beta$ into (\ref{eq_14}) yields
        \begin{equation}
        \label{eq_15}
        \frac{{{\lambda _0}}}{{4\pi }} \le {d_{M - n}} = \frac{r}{{\sqrt {\beta {\rm{\cdot}}{{\left( {1{\rm{ + }}\beta } \right)}^{n - 1}}} }}\;\;\;\;n \ge 1.
        \end{equation}
        Simplifying (\ref{eq_15}), we have
        \begin{equation}
        1 \le n \le \left\lfloor {{\rm{1}} + {{\log }_{1 + \beta }}\frac{{16{\pi ^2}{r^2}}}{{\lambda _0^2\beta }}} \right\rfloor .
        \end{equation}
    \subsection{Proof of Lemma 3}
        According  to (\ref{eq_14}), if ${S_{M - 1}}$ can be unpacked correctly, ${d_{M - 1}}$ needs to satisfy
        \begin{equation}
        {d_{M - 1}} \le \frac{1}{{\sqrt \beta  }}{d_M}
        \end{equation}
        with probability
        \begin{equation}
        {P_{M - 1}} = \frac{{\frac{\theta }{2}{{\left( {\frac{{{d_M}}}{{\sqrt \beta  }}} \right)}^2}}}{{\frac{\theta }{2}{r^2}}} = \frac{{d_M^2}}{{\beta {r^2}}}.
        \end{equation}

        Similarly, the probability that ${S_{M - n}}$ can be unpacked correctly is
        \begin{equation}
        \label{eq_a_46}
        {P_{M - n}} = \frac{1}{{\beta {r^2}\sum\limits_{i = 0}^{n - 1} {{{\left( {\frac{1}{{{d_{M - i}}}}} \right)}^2}} }}\;\;\;\;\;1 \le n \le {n_0} - 1.
        \end{equation}

        In particular, when the current $M-1$ packets are all unpacked except of ${S_M}$,
        since the perfect SIC ignores the interference cancellation residual and noise,
        ${S_M}$ can definitely be unpacked successfully, i.e. ${P_M} = 1$.

        Among the $M$ collision packets ($M \le {n_0}$),
        the probability that the first $Q$ packets can be successfully unpacked is
        \begin{equation}
        P(Q,M) = \prod\limits_{i = 1}^Q {{P_i}} .
        \end{equation}

        Then, the expected probability that one of the $M$ collision packets ($M \le {n_0}$)
        can be successfully unpacked is
        \begin{equation}
        \label{eq_a_48}
        \bar P(M) = \frac{1}{M}\sum\limits_{i = 1}^M {iP({\rm{i,}}M)}  = \frac{1}{M}\sum\limits_{i = 1}^M {\prod\limits_{j = 1}^i {i{P_j}} } .
        \end{equation}

        Combining equations(\ref{eq_a_46}) and (\ref{eq_a_48}),
        we have equation (\ref{eq_L3}).
    \subsection{Proof of Lemma 4}
        In Lemma \ref{lemma_2}, the perfect SIC is considered that can completely
        eliminate the interference caused by the unpacked signal.
        However, in practice there will be interference brought by residual components and noise.
        The power of interference is
        \begin{equation}
        \label{eq_34}
        {C_i} = \xi {S_i} + {N_i},
        \end{equation}
        where $\xi$ is the residual coefficient of interference cancellation and ${N_i}$ is
        the power of additive white Gaussian noise.
        In the case of imperfect SIC,
        the conditions that all $M$ collision data packets can be successfully unpacked are
        \begin{equation}
        \label{eq_35}
        \left\{ \begin{array}{l}
        SIN{R_1} = \frac{{{S_1}}}{{{S_2} + {S_3} + {S_4} +  \ldots  + {S_M} + {N_0}}} \ge \beta ,\\
        SIN{R_2} = \frac{{{S_2}}}{{{C_1} + {S_3} + {S_4} +  \ldots  + {S_M} + {N_0}}} \ge \beta ,\\
        SIN{R_3} = \frac{{{S_3}}}{{{C_1} + {C_2} + {S_4} +  \ldots  + {S_M} + {N_0}}} \ge \beta ,\\
         \ldots \\
        SIN{R_{M - 1}} = \frac{{{S_{M - 1}}}}{{{C_1} + {C_2} +  \ldots  + {C_{M - 2}} + {S_M} + {N_0}}} \ge \beta ,\\
        SIN{R_M} = \frac{{{S_M}}}{{{C_1} + {C_2} +  \ldots  + {C_{M - 2}} + {C_{M - 1}} + {N_0}}} \ge \beta ,
        \end{array} \right.
        \end{equation}
        where $N_0$ is the power of additive white Gaussian noise in the environment.

        Substituting (\ref{eq_34}) into (\ref{eq_35}), we have
        \begin{equation}
        \label{eq_36}
        \left\{ \begin{array}{l}
        SIN{R_1} = \frac{{{S_1}}}{{{S_2} + {S_3} + {S_4} +  \ldots  + {S_M} + N}} \ge \beta ,\\
        SIN{R_2} = \frac{{{S_2}}}{{\xi {S_1} + {S_3} + {S_4} +  \ldots  + {S_M} + N}} \ge \beta ,\\
        SIN{R_3} = \frac{{{S_3}}}{{\xi {S_1} + \xi {S_2} + {S_4} +  \ldots  + {S_M} + N}} \ge \beta ,\\
         \ldots \\
        SIN{R_{M - 1}} = \frac{{{S_{M - 1}}}}{{\xi {S_1} + \xi {S_2} +  \ldots  + \xi {S_{M - 2}} + {S_M} + N}} \ge \beta ,\\
        SIN{R_M} = \frac{{{S_M}}}{{\xi {S_1} + \xi {S_2} +  \ldots  + \xi {S_{M - 2}} + \xi {S_{M - 1}} + N}} \ge \beta ,
        \end{array} \right.
        \end{equation}
        where $N$ is the power sum of the noise in the environment and the noise caused by imperfect cancellation.

        Therefore, although the closed-form solution of the maximum number can't be obtained,
        it is upper bounded by $\left\lfloor {2 + {{\log }_{1 + \beta }}\frac{{16{\pi ^2}{r^2}}}{{\lambda _0^2\beta }}} \right\rfloor $ in the case of perfect SIC, and is lower bounded by 1 in the algorithm without SIC.

\begin{IEEEbiography}[{\includegraphics[width=1.1in,height=1.25in,clip,keepaspectratio]{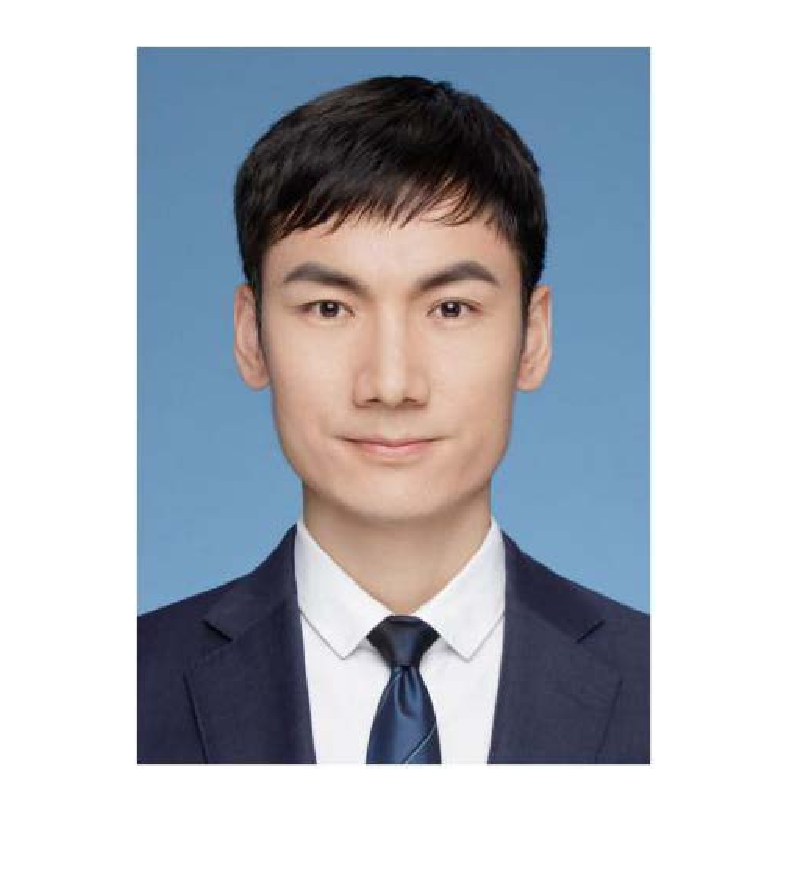}}]
	{Zhiqing Wei}
received his B.E. and Ph.D. degrees from the Beijing University of Posts and Telecommunications (BUPT) in 2010 and 2015. Now he is an associate professor at BUPT.
He was granted the Exemplary Reviewer of IEEE Wireless Communications Letters in 2017, the Best Paper Award of WCSP 2018.
His research interest is the performance analysis and optimization of mobile ad hoc networks.
\end{IEEEbiography}

\begin{IEEEbiography}[{\includegraphics[width=1.1in,height=1.25in,clip,keepaspectratio]{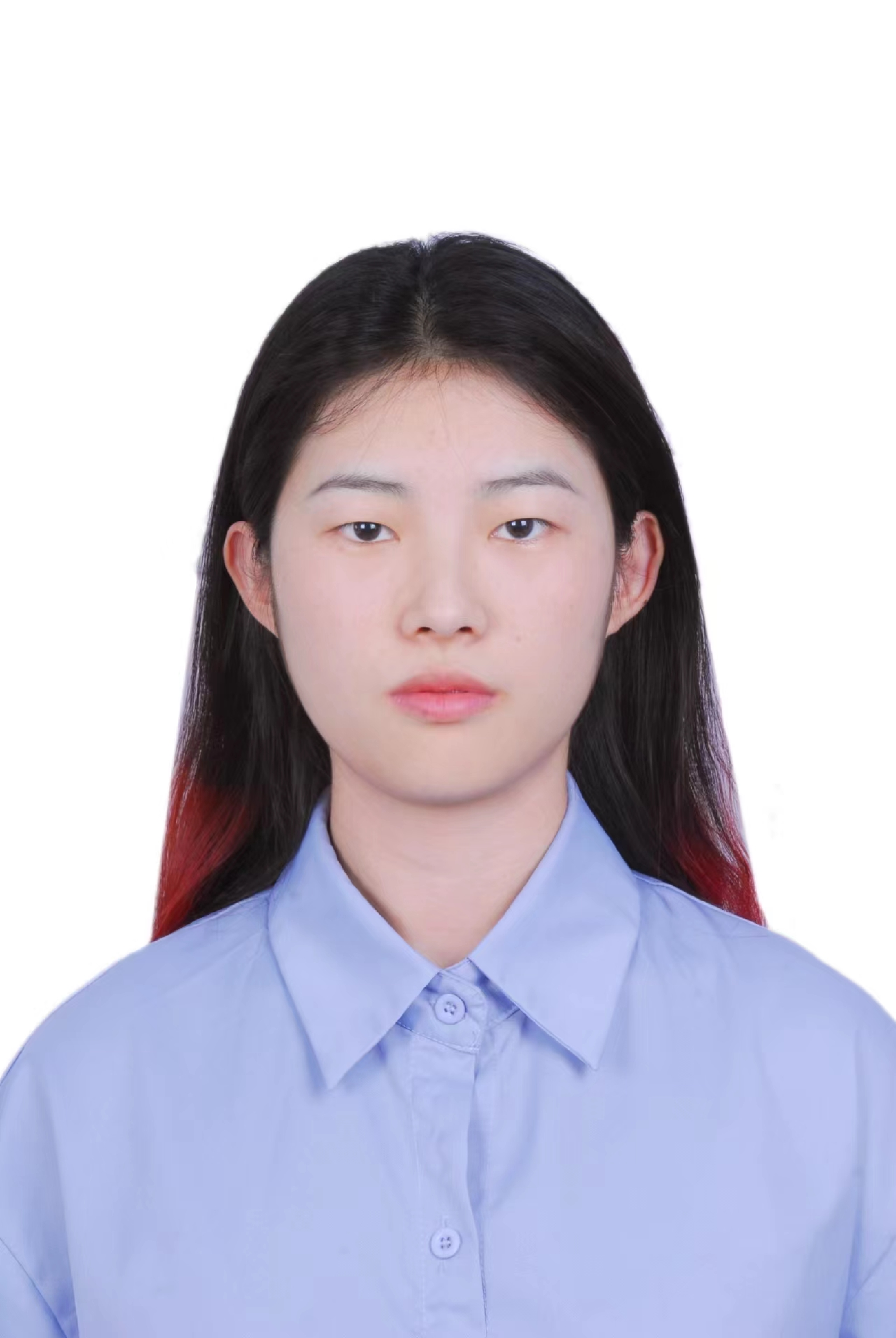}}]
	{Yueyue Liang}
	received the B.S. and M.S. degrees from the Beijing University of Posts and Telecommunications (BUPT),
in 2020 and 2023, respectively. Her research interests include neighbor discovery and relative positioning of wireless ad hoc networks.
\end{IEEEbiography}

\begin{IEEEbiography}[{\includegraphics[width=1.1in,height=1.25in,clip,keepaspectratio]{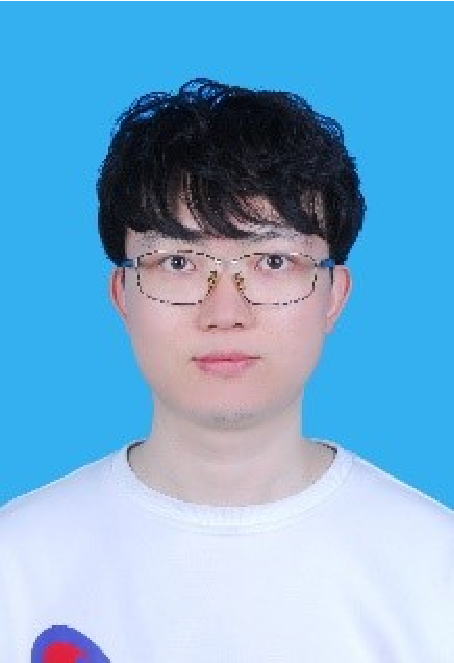}}]
	{Zeyang Meng}
	(Graduate Student Member, IEEE)) received the B.E. degree from Beijing University of Posts and Telecommunications (BUPT), Beijing, China, in 2020. He is working on the Ph.D. degree in BUPT. His research interest is the performance analysis of intelligent machine networks.
\end{IEEEbiography}

\begin{IEEEbiography}[{\includegraphics[width=1.1in,height=1.25in,clip,keepaspectratio]{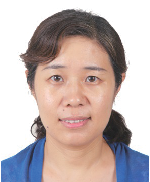}}]{Zhiyong Feng}
is a senior member of IEEE and
a full professor. She is the director of the Key
Laboratory of Universal Wireless Communications,
Ministry of Education. She holds B.S., M.S., and
Ph.D. degrees in Information and Communication
Engineering from Beijing University of Posts and
Telecommunications (BUPT), Beijing, China. She
is a technical advisor of NGMN, the editor of IET
Communications, and KSII Transactions on Internet
and Information Systems, the reviewer of IEEE
TWC, IEEE TVT, and IEEE JSAC. She is active in
ITU-R, IEEE, ETSI and CCSA standards. Her main research interests include
wireless network architecture design and radio resource management in 5th
generation mobile networks (5G), spectrum sensing and dynamic spectrum
management in cognitive wireless networks, universal signal detection and
identification, and network information theory.
\end{IEEEbiography}

\begin{IEEEbiography}[{\includegraphics[width=1.1in,height=1.25in,clip,keepaspectratio]{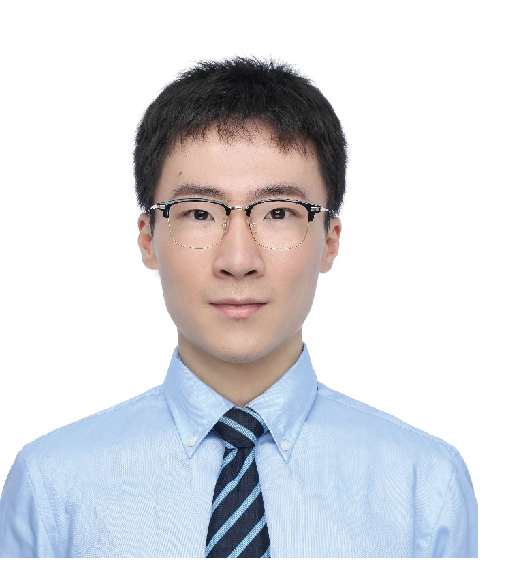}}]{Kaifeng Han} (Member, IEEE) is a senior engineer in the China Academy of Information and Communications Technology (CAICT). Before that, he obtained his Ph.D. degree from the University of Hong Kong in 2019, and the B.Eng. (first-class hons.) from the Beijing University of Posts and Telecommunications and Queen Mary University of London in 2015, all in electrical engineering. His research interests focus on integrated sensing and communications, wireless AI for 6G. He is funded by China Association for Science and Technology (CAST) Young Elite Scientists Sponsorship Program. He received one best paper award and published 40+ research papers in international conferences and journals.
\end{IEEEbiography}

\begin{IEEEbiography}[{\includegraphics[width=1.1in,height=1.25in,clip,keepaspectratio]{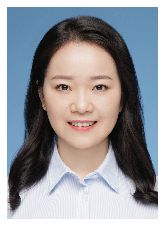}}]{Huici Wu} (Member, IEEE) received the Ph.D degree from Beijing University of Posts and Telecommunications (BUPT), Beijing, China, in 2018. From 2016 to 2017, she visited the Broadband Communications Research (BBCR) Group, University of Waterloo, Waterloo, ON, Canada. She is now an Associate Professor at BUPT. Her research interests are in the area of wireless communications and networks, with current emphasis on collaborative air-to-ground communication and wireless access security.
\end{IEEEbiography}

\ifCLASSOPTIONcaptionsoff
\newpage
\fi

\end{document}